\documentclass[a4paper,UKenglish,cleveref,autoref]{lipics-v2021}

\pdfoutput=1
\hideLIPIcs

\usepackage{xspace}
\usepackage{csquotes}
\usepackage[linesnumbered,ruled,vlined]{algorithm2e}
\usepackage{tikz}

\title{Competitive Query Minimization for Stable Matching with One-Sided Uncertainty}

\author{Evripidis Bampis}{Sorbonne Université, CNRS, LIP6, Paris, France}{evripidis.bampis@lip6.fr}{https://orcid.org/0000-0002-4498-3040}{Partially funded by the grant ANR-19-CE48-0016 from the French National Research Agency (ANR).}
\author{Konstantinos Dogeas}{Department of Computer Science, Durham University, Durham, United Kingdom}{}{https://orcid.org/0009-0001-1528-3221}{}
\author{Thomas Erlebach}{Department of Computer Science, Durham University, Durham, United Kingdom}{thomas.erlebach@durham.ac.uk}{https://orcid.org/0000-0002-4470-5868}{}
\author{Nicole Megow}{Faculty of Mathematics and Computer Science, University of Bremen, Bremen, Germany}{nicole.megow@uni-bremen.de}{https://orcid.org/0000-0002-3531-7644}{Supported by DFG grant no.\ 547924951.}
\author{Jens Schl{\"o}ter}{Faculty of Mathematics and Computer Science, University of Bremen, Bremen, Germany}{jschloet@uni-bremen.de}{https://orcid.org/0000-0003-0555-4806}{}
\author{Amitabh Trehan}{Department of Computer Science, Durham University, Durham, United Kingdom}{amitabh.trehan@durham.ac.uk}{https://orcid.org/0000-0002-2998-0933}{}

\authorrunning{E.~Bampis, K.~Dogeas, T.~Erlebach, N.~Megow, J.~Schlöter and A.~Trehan}

\Copyright{Evripidis Bampis, Konstantinos Dogeas, Thomas Erlebach, Nicole Megow, Schl{\"o}ter, Amitabh Trehan}



\relatedversion{An extended abstract of this paper appears in the proceedings of the International Conference on Approximation Algorithms for Combinatorial Optimization Problems (APPROX 2024).} 

\funding{This research was supported by EPSRC grant EP/S033483/2.}

\nolinenumbers 

\newcommand{\comparison}{comparison\xspace}

\newcommand{\prefer}{\mathit{prefer}}
\newcommand{\topq}{\mathit{top}}
\newcommand{\intq}{\mathit{intq}}

\newcommand{\mtrue}{\mathsf{true}}
\newcommand{\mfalse}{\mathsf{false}}

\newcommand{\OPT}{\mathrm{OPT}}
\newcommand{\ALG}{\mathrm{ALG}}

\DeclareMathOperator{\EX}{\mathbb{E}}

\begin{document}

\maketitle

\begin{abstract}
    We study the two-sided stable matching problem with one-sided uncertainty for two sets of agents $A$ and $B$, with
    equal cardinality. Initially, the preference lists of the agents in $A$ are given but the preferences of the agents in $B$ are unknown. An algorithm can make queries to reveal information about the preferences of the agents in $B$.
    We examine three query models: comparison queries, interviews,
    and set queries. Using competitive analysis, our aim is to
    design algorithms that minimize the number of queries required
    to solve the problem of finding a stable matching or verifying
    that a given matching is stable (or stable and optimal for
    the agents of one side). We present
    various upper and lower bounds on the best possible competitive ratio
    as well as results regarding
    the complexity of the offline problem of determining the
    optimal query set given full information.
    \end{abstract}

    \section{Introduction}
    \label{sec:intro}
    
    In the classical two-sided stable matching problem, we are given two disjoint sets $A$ and~$B$ of agents (often referred to as men and women) of equal cardinality~$n$. Each agent has a complete preference list over the agents of the other set. The task is to find a {\em stable} matching, i.e., a one-to-one allocation in which no two agents prefer to be matched to each other rather than to their current matching partners. This problem has applications in numerous
    allocation markets, e.g.,
    university admission, residency markets, distributed internet services, etc.
    Since its introduction by Gale and Shapley~\cite{GaleShapley1962} this problem has been widely studied in different variants from both practical and theoretical perspectives;
    we refer to the books~\cite{DBLP:books/daglib/0066875,RothSotomayor_book1990,DBLP:books/ws/Manlove13}.
    
    While the majority of the literature assumes full information about the preference lists, this may not be realistic in large matching markets. It might be impractical or too costly and not even necessary to gather the complete preferences.
    Hence, different models for uncertainty in the preferences have received attention in the past decade~\cite{DBLP:journals/ai/AzizBHR19,DBLP:journals/algorithmica/AzizBGHMR20,DrummondB13,DBLP:conf/aaai/DrummondB14,DBLP:journals/jet/EhlersM15,haeringer2019,HaeringerI14,DBLP:conf/sigecom/RastegariCIIL14,DBLP:conf/sigecom/RastegariCIL13}. Many of these works rely on probabilistic models and guarantees. This may not be appropriate for applications in which no (correct) distributional information
    is available, e.g.\ in one-time markets. Further, one might ask for guaranteed properties such as stability and optimality instead of probabilistic ones.
    
    A different way of handling uncertainty in the preferences
    is to allow an algorithm to make queries to learn about the unknown preferences. Various types of queries (in terms of both input and output) are conceivable, with one example being {\em interview queries}~\cite{DrummondB13,DBLP:conf/aaai/DrummondB14,DBLP:conf/sigecom/RastegariCIIL14,DBLP:conf/sigecom/RastegariCIL13}.
    Here one asks for a query sequence
        where a query corresponds to an interview between two potential matching partners and the outcome is the placement of the interview partners in each other's preference list among all other candidates that she has interviewed so far. Hence, if an agent has several such interviews then she finds out her preference order over all these candidates.
    
    In this paper we investigate various query models for stable matching problems with {\em one-sided} uncertainty in the preferences.
    We assume that initially only the preference lists
    of one side,~$A$, are known but the preference lists of the other side, $B$, are unknown. 
    Applications include allocations between groups of different
    seniority or when  preferences shall be kept private; see also~\cite{haeringer2019,HaeringerI14,HaeringerI10}. 
    For illustration consider, e.g., pairing new staff with mentors or new PhD students with supervisors as part of the onboarding. 
    New staff can be asked to provide a full preference list of mentors based on information about the available mentors that can be made accessible with little effort, while requiring mentors to rank potential mentees might be considered too burdensome for senior staff due to other significant time commitments.
    
    We consider three types of queries to gain information about the preferences,
    namely {\em (i)~\comparison queries} that reveal for an agent $b\in B$ and a pair of agents from~$A$ which one~$b$ prefers, {\em (ii) set queries} that reveal for an agent $b\in B$ and a subset~$S \subseteq A$ the agent in~$S$ that $b$ prefers most, and {\em (iii) interview queries}.
    
    We study basic problems regarding stability and optimality of matchings using these query models. 
    A stable matching is called $A$-optimal~(resp.~$B$-optimal) if no agent in $A$~(resp.~$B$) prefers a different stable matching over the current one.
    To our knowledge, most existing related work considers worst-case bounds on the absolute number of queries necessary to solve the respective problem; see Further Related Work below for a discussion. For many instances, however, executing such a worst-case number of queries might not be necessary. To also optimize the number of queries on these instances, we analyze our algorithms using \emph{competitive analysis}.
    We say that an algorithm that makes queries until it can output a provably correct answer, e.g., a stable and $A$-optimal matching, is $\rho$-competitive (or $\rho$-query-competitive)
    if it makes at most $\rho$ times as many queries as the minimum possible number of queries that also output a provably correct answer for the given instance. Note that this answer may differ from that of the algorithm, e.g., a different stable matching.
    In this paper, we design upper and lower bounds on the competitive ratios for the above mentioned problems and query models.	
    Our results illustrate that worst-case instances regarding the competitive ratio are very different from the worst-case instances regarding the absolute number of queries. Thus, our lower bounds on the competitive ratio use different instances and our algorithms are designed to optimize on different instances. Indeed, worst-case instances for the absolute number of queries turn out to be \enquote{easy} for competitive analysis as the optimal solution we compare against is very large.
    
    Query-competitive algorithms are often associated with the field of \enquote{explorable uncertainty}. Most previous work considers queries revealing an originally uncertain {\em value}~\cite{BampisDEdLMS21,DurrEMM20,erlebach15querysurvey,halldorsson19sortingqueries,DBLP:conf/stacs/HoffmannEKMR08,kahan91queries,megow17mst,DBLP:conf/ijcai/ErlebachLMS23}, while in this work we query a {\em preference}.
    
    \subparagraph*{Our Contribution.}
    We study the stable matching problem with one-sided uncertainty in the preference lists and give the following main results. Note that we assume that the preferences of the $B$ side are unknown, and $|A|=|B|=n$. We remark that our technically most involved main results are lower bounds on the competitive ratio and hardness results, so the results only get stronger by
    making these assumptions.
    
    In Section~\ref{sec:pair} we focus on {\em \comparison queries}. Firstly, we ask the question of how to verify that a given matching is stable. We show that
    the problem can be solved with a $1$-competitive algorithm. 
    Then we ask how to find a stable matching under one-sided uncertainty. We give a $1$-competitive algorithm that finds a stable matching and, moreover, the solution is provably $A$-optimal. 
     Essentially, we employ the well-known {\em deferred acceptance algorithm}, first analyzed by Gale and Shapley~\cite{GaleShapley1962}, and compare its number of queries carefully with the number of queries that any algorithm needs to verify a stable matching.
    
    A substantially more challenging
    task is to find a $B$-optimal stable matching. Note that a trivial competitive ratio is 
    $O(n^2 \log n)$, as it is possible to obtain the full preferences of each of the $n$ elements
    in $B$ using $O(n \log n)$ queries, and the optimum total number of queries
    is at least $1$. One of our main contributions is a tight bound of $O(n)$.
    To that end, we first show that every algorithm for verifying that a given matching is $B$-optimal and stable requires $\Omega(n)$ queries.
    Then we give an $\mathcal{O}(n)$-competitive algorithm for the problem of finding one.
    This is best possible up to constant factors, which we prove with a matching lower bound
    that also holds for verifying that a given matching is stable and $B$-optimal,
    even for randomized algorithms.

    We complement these results by showing that the offline problem of determining the optimal number of queries for
    finding the $B$-optimal stable matching is NP-hard, and we give an $\mathcal{O}(\log n \log \log n)$-approximation algorithm.
    Here, the \emph{offline} version of a problem is to compute,
    given full information about the preferences of all agents,
    a smallest set of queries with the property that an algorithm making exactly
    those queries has sufficient information to solve the problem
    with one-sided uncertainty.
    
    Section~\ref{sec:interview} discusses interview queries. We show that the bounds on the competitive ratio and hardness results for comparison queries translate to interview queries. We remark that some of these results for interview queries, e.g., a $1$-competitive algorithm for finding an $A$-optimal stable matching, were already proven by~Rastegari et al.~\cite{DBLP:conf/sigecom/RastegariCIL13} and discuss differences to their results in the corresponding section.
    Interestingly, we can use essentially the same techniques as for the \comparison model. This may seem surprising, especially for the lower bounds, as interview queries seem to be more powerful.
    For instance, $n$ interviews are sufficient to determine the precise preference order of an agent~$b\in B$, while we need $\Omega(n \log n)$ \comparison queries to determine $b$'s preference order.
    On the other hand, an instance that can be solved with a single comparison query requires two interviews. In general, we can simulate a comparison query by using two interview queries.
    
    In Section~\ref{sec:set} we discuss the {\em set query} model. While some bounds remain the same as in the other models, e.g., $1$-competitiveness for verifying the stability of a given matching, we show that some
    bounds change drastically. For example, we give an $\mathcal{O}(\log n)$-competitive algorithm for verifying that a given matching is $B$-optimal, which is in contrast to the lower bound of $\Omega(n)$ in the other query models. It remains open whether $\mathcal{O}(1)$-competitive algorithms exist for the problems of finding a stable matching or verifying a $B$-optimal matching with set queries.
    
    \subparagraph*{Further Related Work.}
    In classical
    work on stable matching with queries, the  preferences on both sides can only be accessed via queries, with a query usually either
    asking for the $i$th entry in a preference list or for the rank of a specific element within a preference list~(cf.~e.g.~\cite{DBLP:journals/siamcomp/NgH90}).
    Note that two rank queries are sufficient to simulate a \comparison query, but up to $n-1$ \comparison queries are needed to obtain the information of a single rank query. 
    Thus, existing lower bounds on the necessary number of rank queries in these query models translate to our setting (up to a constant factor), but upper bounds do not necessarily translate.
    Ng and Hirschberg~\cite{DBLP:journals/siamcomp/NgH90}
    showed that $\Theta(n^2)$ such queries are necessary to find or verify a stable matching in the worst case.
    The lower bound of  $\Omega(n^2)$ translates to any type of queries with boolean answers, including \comparison queries \cite{DBLP:journals/geb/GonczarowskiNOR19}.
    Further work on interview queries includes empirical results~\cite{DrummondB13,DBLP:conf/aaai/DrummondB14} and complexity results~\cite{DBLP:conf/sigecom/RastegariCIIL14} on several decision problems under partial uncertainty. We discuss the latter in~\Cref{sec:interview}.
    
    Our setting of one-sided uncertainty and querying uncertain preferences is also related to existing work 
    on online algorithms for \emph{eliciting partial preferences}~\cite{Hosseini00S21,Peters2022,Ma0L21}. 
    These works also consider a setting where the preferences of agents in one of the sets are uncertain but can be determined by using different types of queries.
    In particular,~\cite{Peters2022} also considers the set query model.
    The main difference to our work is that these papers assume that the elements of one set do not have any preferences at all.
    As a consequence, they do not consider stability at all and instead aim at computing pareto-optimal or rank-maximal matchings.

    \section{Preliminaries}
    \label{sec:prelim}
    
    An instance of the {\em two-sided stable matching problem} consists of two disjoint sets $A$ and $B$ of size $|A|=|B|=n$ and complete preference lists: The preference list for each agent $a\in A$ is a total order~$\prec_a$ of~$B$, the preference list of each agent $b\in B$ is a total order $\prec_b$ of~$A$. Here, $a_1 \prec_b a_2$ means that $b$ prefers $a_1$ to $a_2$.
    A matching is a bijection from $A$ to~$B$.
    For a matching $M$, we denote the element of $B$ that is matched to $a\in A$ by $M(a)$, and the element of $A$ that is matched to $b\in B$ by $M(b)$.
    
    Given a matching~$M$, a pair $(a,b)\in A\times B$ is a \emph{blocking pair} in $M$ if $a$ is not matched to $b$ in $M$, $a$ prefers $b$ to $M(a)$,
    and $b$ prefers $a$ to $M(b)$.
    A matching $M$ is called a \emph{stable} matching if there is no blocking pair in $M$.
    
    In their influential paper, Gale and Shapley~\cite{GaleShapley1962} showed
    that a stable matching always exists, and the {\em deferred acceptance algorithm} computes one in $\mathcal{O}(n^2)$ time. In this algorithm, one
    group ($A$ or $B$) proposes matches and the other decides whether to accept or reject each proposal. The algorithm produces a stable matching that is best possible for the group $X$ that proposes (we say {\em $X$-optimal}) and worst possible for the other group: Each element of the group that proposes gets matched to the highest-preference element to which it can be matched in any stable matching, and each element of the other group gets matched to the lowest-preference element to which it can be matched in any stable matching. 
    
    In this paper, we consider the setting of {\em one-sided uncertainty}, where initially only the preference lists of all agents in $A$ are known, but the preference lists of $b\in B$ are unknown. An algorithm can make queries to learn about the preferences of $b\in B$. We distinguish the following types of queries:
    \begin{itemize}
        \item {\em Comparison queries}: For agents $b\in B$ and $a_1,a_2\in A$, the
        query $\prefer(b,a_1,a_2)$
        returns $a_1$ if $b$ prefers $a_1$ to $a_2$ and $a_2$ otherwise. These queries can  also be seen as Boolean queries that return true iff $b$ prefers $a_1$ to $a_2$.
        \item {\em Set queries}: For agents $b\in B$ and any subset $S \subseteq A$, the
        query $\topq(b,S)$
        returns $b$'s most preferred element of~$S$.
        \item  \emph{Interview queries:} For agents $b \in B$ and $a \in A$, an interview query $\intq(b,a)$ reveals the total order of the subset $\{a\} \cup P_b$ defined by $\prec_b$, where $P_b$ is the set of all elements $a' \in A$ for which a query $\intq(b,a')$ has already been executed before the query $\intq(b,a)$.
    \end{itemize}
    
    A \emph{stable matching instance with one-sided uncertainty}
    is given by two sets $A$ and $B$ of size~$n$ and, for each
    agent $a\in A$, a total order $\prec_a$ of the agents in $B$.
    The preferences of the agents in $B$ are initially unknown.
    For a given stable matching instance with one-sided uncertainty,
    we consider the following problems: \emph{finding a stable
    matching}, \emph{finding an $A$-optimal stable matching},
    and \emph{finding a $B$-optimal stable matching}.
    For a given stable matching instance with one-sided uncertainty
    and a matching $M$, we consider the following problems:
    \emph{verifying that $M$ is stable}, \emph{verifying that $M$
    is stable and $A$-optimal}, and \emph{verifying that $M$
    is stable and $B$-optimal}. All problems can
    be considered for each query model. For the verification
    problems, we consider the competitive ratio only for inputs
    where~$M$ is indeed a stable (and $A$- or $B$-optimal)
    matching. If this is not the case,
    the algorithm must detect this, but we do not compare the number of queries it
    makes to the optimum.
    This is because any algorithm may be required
    to make up to $\Omega(n^2)$ comparison or interview queries to detect
    a blocking pair, while the optimum can
    prove its existence
    with a constant number of queries.
    
    It is easy to see that for the optimum, the problem of verifying that a given
    matching $M$ is stable and $A$-optimal ($B$-optimal) is the same as that of
    finding the $A$-optimal ($B$-optimal) stable matching. This implies
    that any lower bound on the number of queries required to verify that
    $M$ is stable and $A$-optimal ($B$-optimal) also applies to the problem of finding
    the $A$-optimal ($B$-optimal) stable matching.
    
    An important concept is the notion of \emph{rotations}, 
    which can be defined as follows (cf.~\cite{DBLP:books/ws/Manlove13}):
    Let a stable matching $M$ be given.
    For an agent $a_i\in A$, let $s_A(a_i)$ denote the most-preferred element $b_j$
    on $a_i$'s preference list such that $b_j$ prefers $a_i$ to her current partner
    $M(b_j)$.
    Note that $s_A(a_i)$ must be lower than $M(a_i)$ in $a_i$'s preference
    list as otherwise $(a_i,s_A(a_i))$ would be a blocking pair.
    Let $\mathrm{next}_A(a_i) = M(s_A(a_i))$.
    Then a rotation (exposed) in $M$ is a sequence
    $
    (a_{i_0} , b_{j_0}), \ldots, (a_{i_{r-1}} , b_{j_{r-1}})
    $
    of pairs such that, for each $k$ ($0 \le k \le r - 1$),
    $(a_{i_k} , b_{j_k}) \in M$ and $a_{i_{k+1}} = \mathrm{next}_A(a_{i_k})$,
    where addition is modulo~$r$.
    The rotation can be viewed as an alternating cycle consisting
    of the matched edges $(a_{i_k},b_{i_k})$ and the unmatched
    edges $(a_{i_k},b_{i_{k+1}})$ (for $0\le k\le r-1$).
    We refer to an edge $(a,s_A(a))$ as a \emph{rotation edge} or \emph{$r$-edge} as it can
    potentially be part of a rotation. Note that every vertex
    $a\in A$ is incident with at most one $r$-edge.
    
    Given a rotation $R$
     in a stable matching $M$, we can construct a stable matching~$M'$ from~$M$ by removing all edges that are part of $R$ and $M$ and adding all $r$-edges that are part of~$R$. We refer to this 
     as \emph{applying} a rotation.
    Observe that no agent in $B$ is worse off in $M'$ than in~$M$, and some agents in $B$ prefer $M'$ to~$M$.
    The following
    has been shown.
    \begin{lemma}[Lemma~2.5.3 in Gusfield and Irving \cite{DBLP:books/daglib/0066875}%
    ]
    \label{lem:rotations-for-optimal}
        If $M$ is any stable matching other than the $B$-optimal stable matching, then there is at least one rotation exposed in~$M$.
    \end{lemma}

    \section{Stable Matching with Comparison Queries}
    \label{sec:pair}
    
    In this section, we consider the comparison query model with one-sided uncertainty. We first discuss our results on the problems of verifying that a given matching is stable and finding an $A$-optimal matching, before moving on to our main results regarding the competitive ratio for finding/verifying a $B$-optimal matching.
	Finally, we briefly consider the variation with two-sided uncertainty and
	give tight bounds for the problem of verifying a stable matching in that model.
    
    \subsection{Verifying That a Given Matching Is Stable}
    \label{subsec:pair-verifyStable}
    
In this section, we consider the {\em verification problem} where we are given a matching $M$ and our task is to verify that $M$ is indeed stable. We give a $1$-competitive algorithm.
As argued in the previous section, we only care about the competitive ratio if the given matching $M$ is indeed stable.
If the given matching $M$ is not stable, the algorithm must detect this, but its number of queries can be arbitrarily much larger than the optimal number of queries for detecting that $M$ is not stable.
In the case of one-sided uncertainty, as we consider it here, a single query is sufficient for the optimum to identify a blocking pair.

The following auxiliary lemma shows that exploiting transitivity cannot reduce the number of \comparison queries
to an agent $b\in B$ if one needs to find out the preference relationship of~$k$ agents
from $A$ to one particular agent from $A$ in $b$'s preference list.

\begin{lemma}\label{lem:transitivenohelp}
Consider two agents $a\in A$, $b\in B$ and assume that
there are $k$ agents $a_1,\ldots,a_k\in A\setminus\{a\}$ for
each of which we want to know whether $b$ prefers that agent to $a$ or not.
Then exactly $k$ \comparison queries to $b$ are necessary and sufficient to obtain~this~knowledge.
\end{lemma}

\begin{proof}
The $k$ queries $\prefer(b,a,a_i)$ for $i=1,\ldots,k$ are clearly
sufficient.
Assume that $k'$ queries to~$b$, for some $k'<k$, are sufficient
to obtain the desired information. Consider the auxiliary
graph $H$ with vertex set $V_H=A$ and an edge
$\{a',a''\}$ for each of those $k'$ queries $\prefer(b,a',a'')$.
As the set $A'=\{a,a_1,a_2,\ldots,a_k\}$ has $k+1$ vertices
and $H$ has fewer than $k$ edges, the set $A'$ intersects at
least two different connected components of~$H$. Let
$a_j$, for some $1\le j\le k$, be a vertex that does
not lie in the same component as~$a$. Then the $k'$ queries
do not show whether $b$ prefers $a_j$ to $a$ or not,
which contradicts $k'$ queries~being~sufficient.
\end{proof}

If an algorithm
obtains for agents $x$ and $y$ with $y\neq M(x)$
the information that $M(x) \prec_x y$ (either via a direct query or via transitivity), we say that
the algorithm \emph{relates} $y$ to $M(x)$ for~$x$.
By Lemma~\ref{lem:transitivenohelp}, if the optimum
relates $k$ different elements to $M(x)$ for~$x$,
it needs to make $k$ queries to~$x$.
A pair $(x,y)$ with $y\neq M(x)$ such that the optimum
relates $y$ to $M(x)$ for $x$ is called a \emph{relationship pair} (for~$x$).
Lemma~\ref{lem:transitivenohelp} implies the following.

\begin{corollary}\label{cor:relinst}
The total number of relationship pairs (for all agents $x$)
is a lower bound on the number of comparison queries the optimum makes.
\end{corollary}

    \begin{theorem}
    \label{thm:1-comp-verification}
        Given a stable matching instance with one-sided uncertainty and a stable matching~$M$, there is a $1$-competitive algorithm that uses $\sum_{a\in A} |\{b\in B \mid b \prec_a  M(a)\}|$ queries for verifying that $M$ is stable in the \comparison query model.
    \end{theorem}
    
\begin{proof}
	Since the preferences of agents on the $A$-side are not uncertain, for each $(a,b) \notin M$, we already know whether $M(a) \prec_a b$. If $M(a) \prec_a b$, then we do not have to execute any queries to show that $(a,b) \notin M$ is not a blocking pair. Otherwise, every feasible query set has to prove $M(b) \prec_b a$.
	Therefore, for each element $b\in B$, there is a uniquely determined number $n_b$ of elements of $A$
	that any solution (including the optimum) must relate to $M(b)$ for $b$.
	Let $K=\sum_{b\in B} n_b$ be the resulting number of relationship pairs.

	Our algorithm simply
	queries $\prefer(b,M(b),a)$ for every pair $(a,b)\notin M$ for which $b\prec_a M(a)$. These are exactly $K$ queries.
    As the total number of relationship pairs is~$K$, 
	the optimum must also make $K$ queries
	(Corollary~\ref{cor:relinst}).
	Hence, our algorithm~is~$1$-competitive.
\end{proof}

The proof of Theorem~\ref{thm:1-comp-verification} implies that the stable matching
that maximizes the number of queries that are required to prove stability is
the $ B$-optimal matching.

    \begin{corollary}
        \label{cor:Bopt:stability:lb}
        The number of comparison queries needed to verify that the $B$-optimal
        matching is stable is $\max_{ M\mbox{ stable}} \sum_{a\in A} |\{b\in B \mid b \prec_a  M(a)\}|$.
    \end{corollary}

    \subsection{Finding an \texorpdfstring{$A$}{A}-Optimal Stable Matching} 
    \label{subsec:pair-Aoptimal}
    
    We obtain the following positive result by adapting the classical deferred acceptance algorithm~\cite{GaleShapley1962} with $A$ making the proposals. 
    
    \begin{theorem}
    \label{th:pairwiseAoptimal}%
    For a given stable matching instance with one-sided uncertainty, there is a $1$-competitive algorithm for finding a stable matching in the \comparison query model. The algorithm actually finds an $A$-optimal stable matching.
    \end{theorem}

    \begin{proof}
We utilize the classical deferred acceptance algorithm
\cite{GaleShapley1962} where $A$
makes the proposals, and we assume the reader's familiarity with it. An unmatched agent $a\in A$
makes a proposal to their preferred agent $b\in B$ by whom it has never been rejected. If $b$ is unmatched, then $b$ accepts the proposal and $a$ and $b$ get matched. If $b$ is
currently matched to some $a'\in A$, the algorithm
makes a query $\prefer(b,a,a')$. If the query result
is that $b$ prefers $a$ to $a'$, then $b$ accepts
$a$'s proposal and becomes matched to $a$ while
$a'$ becomes unmatched. Otherwise, $b$ rejects
the proposal and remains matched to~$a'$. The algorithm terminates if all agents in $A$ are matched or if every unmatched agent in $A$ has been declined by all agents in~$B$.

We show that this algorithm makes the minimum
possible number of \comparison~queries.

We execute the deferred acceptance algorithm
with $A$ as the proposers, so it produces
an $A$-optimal stable matching. Consider
an arbitrary agent~$b\in B$. Assume that
$b$ gets  matched to $a\in A$ in the stable
matching. Let $A_b=\{a,a_1,a_2,\ldots,a_{k_b}\}$ (for
some $0\le k_b <n$) be the set of agents of $A$
that proposed to $b$ during the execution
of the algorithm. Note that $|A_b|=k_b+1$ and
the algorithm has executed $k_b$ queries to~$b$, each
for two agents of $A_b$ (the first agent of
$A$ that proposed to $b$ did not require a query).
Observe that each of $a_1,\ldots,a_{k_b}$ gets
matched with an agent of $B$ that they rank
strictly lower than $b$ in the final matching.

We claim that no stable matching can be identified
without making at least $k_b$ queries to~$b$.
Let $M'$ be an arbitrary stable matching.
Note that $b$ rates $M'(b)$ at least as highly
as~$a$, because $M$ is the worst possible
matching for~$B$. Furthermore, for each $a_i$
with $1\le i\le k_b$, we have that $a_i$
rates $b$ strictly higher than $M'(a_i)$
because $M$ is $A$-optimal and $a_i$ rates $M(a_i)$
strictly lower than $b$.
Thus, for none of the pairs $(a_i,b)$ for $1\le i \le k_b$
to be a blocking pair, the queries of any
optimal query set must establish that $b$ rates
$M'(b)$ more highly than every $a_i$ for $1\le i\le k_b$.
This can only be achieved with at least $k_b$ queries.

The same argument applies to each $b\in B$, so we
have that both the optimal number of queries and
the number of queries made by the algorithm
are equal to $\sum_{b\in B} k_b$.
\end{proof}

The proof of Theorem~\ref{th:pairwiseAoptimal} implies
that, for the $A$-optimal stable matching $M$, the optimal number
of queries to prove that $M$ is stable
equals the optimal number of queries to prove that $M$ is
stable and $A$-optimal. Hence, proving optimality comes in this case for free.
    
    \subsection{Finding a \texorpdfstring{$B$}{B}-Optimal Stable Matching}
    \label{subsec:pair-Boptimal}

    The problem of finding a $B$-optimal stable matching is substantially more challenging
	in general.
	For the special case where all $A$-side preference lists are equivalent, however,
	there exists a $1$-competitive algorithm:

	\begin{observation}
	\label{obs:stability-equal-B-prefs}%
	If all agents of $A$ have the same preference list, then there is a $1$-competitive algorithm that uses $\frac{n^2-n}{2}$ queries for finding a stable matching under one-sided uncertainty.
	\end{observation}

\begin{proof}
Let $B=\{b_1,\ldots,b_n\}$ be indexed by the preference list of the elements in $A$, i.e., $b_1$ is the first choice of the elements in $A$ and $b_i$ is the $i$th choice of the elements in $A$.

Let $a^*_1$ be the (initially unknown) first choice of $b_1$ and, for $i > 1$, let $a_i^*$ denote the first choice of $b_i$ among the elements of $A \setminus \{a^*_1,\ldots,a^*_{i-1}\}$. We can inductively argue that every stable matching must match $b_i$ to $a_i^*$ for all $i \in [n]$.

Thus, every algorithm (including the optimal solution) has to find the first choice of $b_i$ in $A \setminus \{a^*_1,\ldots,a^*_{i-1}\}$ for all $i \in [n]$. For each $b_i$ this requires a minimum of $|A \setminus \{a^*_1,\ldots,a^*_{i-1}\}|-1 = n-i$ queries, as each query can exclude at most one element from being the first choice of $b_i$. This implies that every algorithm needs at least $\sum_{i=1}^n (i-1) = \frac{n^2-n}{2}$ queries.

The following algorithm matches this lower bound. This is essentially the deferred acceptance algorithm used in the proof of~\Cref{th:pairwiseAoptimal} but considers the agents of $B$ in a specific order: (i) Iterate through $B$ in order of increasing indices (ii)
For each $b_i$, determine the first choice among the not yet matched elements of $A$ and match $b_i$ to that choice. For each $b_i$ this requires at most $n-i$ queries and, thus, a total of $\sum_{i=1}^n (i-1) = \frac{n^2-n}{2}$ queries.
\end{proof}

    For arbitrary instances, 
    we first describe an algorithm that is $\mathcal{O}(n)$-competitive.
    Complementing this result, we then show that every (randomized) online algorithm has competitive ratio at least $\Omega(n)$ for finding a $B$-optimal stable matching.
    Finally, we show that the offline problem of determining the optimal number of queries for computing a $B$-optimal stable matching is NP-hard and give an $\mathcal{O}(\log n \log \log n)$-approximation.  
    
    \subsubsection{Algorithm for Computing a \texorpdfstring{$B$}{B}-Optimal Stable Matching}
    \label{subsubsec:pair-B-algo}

    We first consider the problem of verifying that a given stable $B$-optimal matching is indeed stable and $B$-optimal.
    An algorithm for this problem needs to prove that $M$ has no blocking pair and that no alternating cycle with respect to $M$ is a rotation.
    For each potential blocking pair $(a,b)$ that cannot be ruled out
    because of $a$'s preferences, such an algorithm has to prove that it is not a blocking pair
    using a suitable query to $b$ as discussed in~\Cref{subsec:pair-verifyStable}.
    
    The more involved part is proving that $M$ is $B$-optimal. By~\Cref{lem:rotations-for-optimal}, $M$ is $B$-optimal if and only if it does not expose a rotation. Based on the known $A$-side preferences, each edge $(a,b)$ with $M(a) \prec_a b$ could potentially be an $r$-edge. Thus, each cycle that alternates between such edges and edges in $M$ could potentially be a rotation.
    An algorithm that proves $B$-optimality has to prove for each such alternating cycle that at least one non-matching edge $(a,b)$ on that cycle is not an $r$-edge.
    By definition,
     there are two possible ways to prove that an edge $(a,b)$ with $M(a) \prec_a b$ is not an $r$-edge:
    \begin{enumerate}
    \item Query $b$ and find out that $b$ prefers $M(b)$ to $a$. Then, $b$ cannot be $s_A(a)$ as $b$ does not prefer $a$~to~$M(b)$.
    \item Query one $b'$ with $M(a) \prec_a b' \prec_a b$ and find out that
    $b'$ prefers $a$ to $M(b')$.
    Then, $b$ cannot be the most-preferred element in
    $a$'s list that prefers $a$ to her current partner, as $b'$
    has that property and is preferred over~$b$.
    \end{enumerate}
    
    Corollary~\ref{cor:Bopt:stability:lb} gives the optimal number of queries to prove that the matching $M$ is stable, which is a lower bound on the optimal number of queries necessary to prove that $M$ is stable and $B$-optimal.
    Let $Q(M)$ denote this number. However, there exist instances where $Q(M)=0$ and $Q_B(M)>0$ for the optimal number $Q_B(M)$ of queries to prove that $M$ is stable \emph{and} $B$-optimal.
    Consider an instance where all elements of $A$ have distinct first choices and let $M$ denote the matching that matches all elements of $A$ to their respective first choice. Then, there is a realization of $B$-side preference lists such that the matching $M$ is also $B$-optimal. For this realization we have $Q(M)=0$ and $Q_B(M)>0$. 
    This implies that the lower bound of Corollary~\ref{cor:Bopt:stability:lb} is not strong enough for analyzing algorithms that verify $B$-optimality as we cannot prove that such an algorithm makes at most $c \cdot Q(M)$ queries.
    We give another lower bound on the optimal number of queries.
    
    \begin{lemma}
        \label{lem:bopt:numberlb:pw}
        The optimal number of queries for verifying (and thus also for finding)
      the $B$-optimal stable matching is at least $n-1$ for every instance of the stable matching problem with one-sided uncertainty.
    \end{lemma}
    
    \begin{proof}
    Let $M$ be the $B$-optimal stable matching for the given instance.
    For $a\in A$, call a query an $a$-query if it reveals for some $b\in B$ with $b\neq  M(a)$ whether $b$ prefers $a$ to her current
    partner or not.
    We claim that an optimal algorithm
    needs to make at least one $a$-query for every $a\in A$ with at most a single exception.
    Assume for a contradiction that
    the optimal algorithm makes neither an $a$-query nor an $a'$-query for two distinct elements $a,a'\in A$.
    If $a$ prefers $ M(a)$ over $ M(a')$ and $a'$ prefers $ M(a')$ over $ M(a)$,
    then it is impossible to exclude the possibility that
    $(a,M(a)),(a',M(a'))$ is a rotation exposed in~$ M$,
    because the only
    way to prove that $(a, M(a'))$ is not an $r$-edge is via an $a$-query, and similarly for $(a', M(a))$.
    If $a$ prefers $ M(a')$ over $ M(a)$, then an $a$-query to $ M(a')$ is necessary to exclude
    that $(a, M(a'))$ is a blocking pair.
    If $a'$ prefers $ M(a)$ over $ M(a')$, then an $a'$-query to $ M(a)$ is necessary for the
    analogous reason.
    Hence, the claim holds. We note that the $n-1$ queries whose existence is asserted by
    the claim are distinct: A query to some $b\in  B$ cannot be an $a$-query and at the same time an $a'$-query
    for some $a'\neq a$, as the query $\prefer(b,a,a')$ cannot yield previously unknown information about how both
    $a$ and $a'$ compare to $ M(b)$ in $b$'s preference~list.
    \end{proof}
    
    Next, we give an $\mathcal{O}(n)$-competitive algorithm for finding a $B$-optimal matching and analyze it by exploiting the lower bounds on the optimal number of queries of~\Cref{cor:Bopt:stability:lb} and~\Cref{lem:bopt:numberlb:pw}. For pseudocode see~\Cref{alg:compqueries:Bopt}.
    
    \begin{enumerate}
    \item Find an $ A$-optimal matching using the 1-competitive algorithm for $ A$-optimal matchings.
    \item Search for a rotation by asking, for every $a\in A$, the elements of $ B$ that are below
      $ M(a)$ in $a$'s preference list in order of $\prec_a$ whether they prefer $a$ to their current partner, until
      either an $r$-edge is found or we know that $a$ has no $r$-edge.
    \item If a rotation $R$ is found, apply that rotation. The agents $a\in  A\cap R$ then no longer have
      a known $r$-edge as their previous $r$-edge is now their matching edge.
      However, the new $r$-edge partner of such an agent must be further down the preference list of $a$ than the old one.  
      The elements $a\in A\setminus R$
      that had an $r$-edge to an element $b\in  B\cap R$ can no longer be sure that their edge to $b$ is an $r$-edge since $b$ has a new matching partner $M(b)$,
      so $b$ must be asked again whether it prefers the new partner over $a$ when searching for the new $r$-edge of $a$.
      The algorithm then repeats Step $2$ but starts the search for the new rotation edge of an agent $a \in A$ at either the previous rotation edge (if  $a\in A\setminus R$) or at the direct successor of the new $M(a)$ in $\prec_a$ (if $a \in A \cap R$).
    \item When a state is reached where it is known for every $a\in  A$ what its $r$-edge is
      (or that it has no $r$-edge) but the $r$-edges do not form a rotation, the algorithm
      terminates and outputs~$ M$.
    \end{enumerate}

    \begin{algorithm}[t]
	\KwIn{Instance of the stable matching problem with one-sided uncertainty.}
	$M \gets $ $A$-optimal matching computed using~\Cref{th:pairwiseAoptimal} \label{line:Bopt:Aopt}\;
	$N \gets \{a \in A \mid M(a) \text{ is last in} \prec_a\}$ \tcc*[r]{Elements without $r$-edge}
	$\forall a\in A \setminus N\colon p(a) \gets$ first element in $\prec_a$ after $M(a)$\;
	$\forall a\in A \setminus N\colon r(a) \gets \top$ \tcc*[r]{known $r$-edges or $\top$ if $r$-edge still unknown} 
	\ForEach{$a \in A\setminus N$\label{line:Bopt:startloop}}{
		\Repeat{$r(a) \not= \top$ or $a \in N$}{
			$t \gets \prefer(p(a),a,M(p(a)))$ \label{line:Bopt:query}\;
			\If{$t = M(p(a))$}{
				\lIf{$p(a)$ is the last element of $\prec_a$}{
					$N \gets N \cup \{a\}$
				}\lElse{
					$p(a) \gets$ direct successor of $p(a)$ in $\prec_a$\label{line:Bopt:md1}
				}
			}	
			\Else{
				$r(a) \gets p(a)$\tcc*[r]{$r(a)$ and $a$ form an $r$-edge}
			}
		}
	}
	\If{$M$ exposes a rotation $R$\label{line:Bopt:rotation}}{
		$M \gets $ stable matching constructed from $M$ by applying $R$\;
		$N \gets N \cup \{a \in A \cap R \mid M(a) \text{ is last in} \prec_a\}$\;
		$\forall a \in (A\cap R) \setminus N \colon r(a) \gets \top$ and  $p(a) \gets$ first element in $\prec_a$ after $M(a)$\label{line:Bopt:md2}\;
		$\forall a \in (A\setminus R)\setminus N \colon p(a) \gets r(a)$ and $r(a) \gets \top$\;
		Jump to Line~\ref{line:Bopt:startloop}\label{line:Bopt:uc}\;
		
	}
	\Return $M$\;
	\caption{Algorithm to find the $B$-optimal stable matching using comparison queries.}
	\label{alg:compqueries:Bopt}
\end{algorithm} 
    
    \begin{theorem}\label{thm:O(n)-comp-B-optimal}
    Given a stable matching instance with one-sided uncertainty,
    the algorithm is $\mathcal{O}(n)$-competitive for finding a $B$-optimal stable matching using \comparison queries.
    \end{theorem}
    
    \begin{proof}
    Let $\OPT$ denote the number of queries made by an optimal algorithm.
    Since finding any stable matching can never require more queries than finding a $B$-optimal stable matching, \Cref{th:pairwiseAoptimal} implies that the algorithm makes at most $\OPT$ queries in the first step.
            
    We analyze the queries executed \emph{after} the first algorithm step.
    Call a query \emph{good} if it is the first query
    involving a specific combination of an agent $a \in A$ and an agent $p(a) \in B$,
    i.e., the first query of form $\prefer(p(a),a,M(p(a)))$ for that specific combination of $p(a)$ and $a$.
    All other queries are \emph{bad}.
    By definition of good queries, the algorithm makes at most $n^2$ such queries since this is the maximum number of good queries that can exist.
    Since $\OPT\ge n-1$ (\Cref{lem:bopt:numberlb:pw}),  the number of good queries~is~$\mathcal{O}(n)\cdot \OPT$.
    
    Consider the bad queries and a fixed $a \in A$. In the second step of the algorithm, it repeatedly executes queries of the form $\prefer(p(a),a,M(p(a)))$ with $p(a) \in B$ to find out if $(a,p(a))$ is an $r$-edge, starting with the direct successor $p(a)$ of $M(a)$ in $\prec_a$. 
    If $(a,p(a))$ is not an $r$-edge, then the next query partner $p(a)$ for $a$ moves one spot down
    in the list $\prec_a$. This is repeated until the $r$-edge $(a,r(a))$ of $a$ is found or we know
    that $a$ does not have~an~$r$-edge. Here, $r(a)$ refers to the element that forms an $r$-edge with $a$.
    
    If $a$ does not have an $r$-edge, there will be no more queries for $a$ again as all $b \in B$ that are lower than $M(a)$ in the preference list of $a$ prefer their current partner $M(b)$ over $a$ and this partner will only improve during the execution of the algorithm. 
    Otherwise, $a$ will be considered again in the second step of the algorithm only if a rotation was found in the third step. If $a$ is part of the rotation, then $r(a) = p(a)$ will be the new matching partner of $a$ and $p(a)$ will be moved one spot down in $\prec_a$. Only if $a$ is not part of the rotation, $p(a) = r(a)$ remains unchanged by definition of the third step.
    In conclusion, the next query partner $p(a)$ of $a$ moves down one spot in $\prec_a$ after each query for $a$ unless a rotation is found that does not contain $a$. This means that a bad query for $a$ can only occur as the first query for $a$ after a new rotation that does not involve $a$ is found. Thus, each rotation can cause at most $|A| - 2$ bad queries (at least two members of $A$ must be involved in the rotation). Thus, the number of bad queries is at most $(n-2) \cdot n_r$ for the number of applied rotations $n_r$. 
    
    For each applied rotation, at least two agents of $A$ get re-matched to agents of $B$
    that are lower down on their preference lists than their previous matching partner. This increases the lower bound
    on the optimal number of queries to show stability (cf.~\Cref{cor:Bopt:stability:lb}) by at least~$2$. 
    Thus,~\Cref{cor:Bopt:stability:lb} implies $\OPT \ge 2 \cdot n_r$. We can conclude that the number of bad queries is at most $(n-2) \cdot n_r \le \mathcal{O}(n) \cdot \OPT$.
    \end{proof}
    
    \subsubsection{Lower Bound for Computing a \texorpdfstring{$B$}{B}-Optimal
    Matching}
    \label{subsubsec:pair-B-lower}
    
    We give
    a lower bound of $\Omega(n)$ on the competitive ratio 
    for finding a $B$-optimal stable matching with comparison queries.
    This implies that the result of~\Cref{thm:O(n)-comp-B-optimal} is, asymptotically, best-possible.
    Further, the lower bound also holds for
    verifying that a given matching is~$B$-optimal.
    
    \begin{theorem}
        \label{thm:B-opt:det-lb}
        In the \comparison query model, every deterministic or randomized online algorithm for finding a $B$-optimal stable matching in a stable matching instance with one-sided uncertainty
        has competitive ratio~$\Omega(n)$.
    \end{theorem}
    
    \begin{proof}
        We first show the statement for deterministic algorithms.
        Consider the following instance (cf.~Fig.~\ref{fig:BoptLB}) with two sets of agents $A = \{ a_0, \dots, a_{n-1} \}$ and $B = \{ b_0, \dots, b_{n-1} \}$, and assume $n/2$ to be even. If this is not the case, then the constant factor in the lower bound will be slightly~worse.
    
    \begin{figure*}
    \centering
\begin{tikzpicture}[scale=0.92,transform shape]
	\node[black, draw, fill, circle, label=left:${a_0}$] (a0) at (0,-.6*0) {};
	\node[black, draw, fill, circle, label=right:${b_0}$] (b0) at (5,-.6*0) {};
	\draw[thick] (a0) -- (b0);
	
	\node[black, draw, fill, circle, label=left:${a_1}$] (a1) at (0,-.6*1) {};
	\node[black, draw, fill, circle, label=right:${b_1}$] (b1) at (5,-.6*1) {};
	\draw[thick] (a1) -- (b1);
	
	\node[black, draw, fill, circle, label=left:${a_2}$] (a2) at (0,-.6*2) {};
	\node[black, draw, fill, circle, label=right:${b_2}$] (b2) at (5,-.6*2) {};
	\draw[thick] (a2) -- (b2);
	
	\node[black, draw, fill, circle, label=left:${a_3}$] (a3) at (0,-.6*3) {};
	\node[black, draw, fill, circle, label=right:${b_3}$] (b3) at (5,-.6*3) {};
	\draw[thick] (a3) -- (b3);
	
	\node[black, draw, fill, circle, label=left:${a_4}$] (a4) at (0,-.6*4) {};
	\node[black, draw, fill, circle, label=right:${b_4}$] (b4) at (5,-.6*4) {};
	\draw[thick] (a4) -- (b4);
	
	\node[black, draw, fill, circle, label=left:${a_5}$] (a5) at (0,-.6*5) {};
	\node[black, draw, fill, circle, label=right:${b_5}$] (b5) at (5,-.6*5) {};
	\draw[thick] (a5) -- (b5);
	
	\node[black, draw, fill, circle, label=left:${a_6}$] (a6) at (0,-.6*6) {};
	\node[black, draw, fill, circle, label=right:${b_6}$] (b6) at (5,-.6*6) {};
	\draw[thick] (a6) -- (b6);
	
	\node[black, draw, fill, circle, label=left:${a_7}$] (a7) at (0,-.6*7) {};
	\node[black, draw, fill, circle, label=right:${b_7}$] (b7) at (5,-.6*7) {};
	\draw[thick] (a7) -- (b7);
	
	\node[black, draw, fill, circle, label=left:${a_8}$] (a8) at (0,-.6*8) {};
	\node[black, draw, fill, circle, label=right:${b_8}$] (b8) at (5,-.6*8) {};
	\draw[thick] (a8) -- (b8);
	
	\node[black, draw, fill, circle, label=left:${a_9}$] (a9) at (0,-.6*9) {};
	\node[black, draw, fill, circle, label=right:${b_9}$] (b9) at (5,-.6*9) {};
	\draw[thick] (a9) -- (b9);
	
	\node[black, draw, fill, circle, label=left:${a_{10}}$] (a10) at (0,-.6*10) {};
	\node[black, draw, fill, circle, label=right:${b_{10}}$] (b10) at (5,-.6*10) {};
	\draw[thick] (a10) -- (b10);
	
	\node[black, draw, fill, circle, label=left:${a_{11}}$] (a11) at (0,-.6*11) {};
	\node[black, draw, fill, circle, label=right:${b_{11}}$] (b11) at (5,-.6*11) {};
	\draw[thick] (a11) -- (b11);
	
	\draw[dashed] (a6) -- (b11);
	\draw[dashed] (a7) -- (b11);
	\draw[dashed] (a8) -- (b11);
	\draw[dashed] (a9) -- (b11);
	\draw[dashed] (a10) -- (b11);
	\draw[dashed] (a11) -- (b11);
	
	\node[] (a0.list) at (-3, -.6*0) {$(b_0, b_6, b_7, b_8, b_9, b_{10}, b_{11}, \ast)$};
	\node[] (a1.list) at (-3, -.6*1) {$(b_1, b_6, b_7, b_8, b_9, b_{10}, b_{11}, \ast)$};
	\node[] (a2.list) at (-3, -.6*2) {$(b_2, b_6, b_7, b_8, b_9, b_{10}, b_{11}, \ast)$};
	\node[] (a3.list) at (-3, -.6*3) {$(b_3, b_6, b_7, b_8, b_9, b_{10}, b_{11}, \ast)$};
	\node[] (a4.list) at (-3, -.6*4) {$(b_4, b_6, b_7, b_8, b_9, b_{10}, b_{11}, \ast)$};
	\node[] (a5.list) at (-3, -.6*5) {$(b_5, b_6, b_7, b_8, b_9, b_{10}, b_{11}, \ast)$};
	
	\node[] (a6.list) at (-4.2, -.6*6) {$(b_6, b_{11}, \ast)$};
	\node[] (a7.list) at (-4.2, -.6*7) {$(b_7, b_{11}, \ast)$};
	\node[] (a8.list) at (-4.2, -.6*8) {$(b_8, b_{11}, \ast)$};
	\node[] (a9.list) at (-4.2, -.6*9) {$(b_9, b_{11}, \ast)$};
	\node[] (a10.list) at (-4.15, -.6*10) {$(b_{10}, b_{11}, \ast)$};
	\node[] (a11.list) at (-4.45, -.6*11) {$(b_{11}, \ast)$};
	
	\node[] (b0.list) at (6.6, -.6*0) {$(a_1, a_0, \ast)$};
	\node[] (b1.list) at (6.6, -.6*1) {$(a_0, a_1, \ast)$};
	\node[] (b2.list) at (6.6, -.6*2) {$(a_3, a_2, \ast)$};
	\node[] (b3.list) at (6.6, -.6*3) {$(a_2, a_3, \ast)$};
	\node[] (b4.list) at (6.6, -.6*4) {$(a_5, a_4, \ast)$};
	\node[] (b5.list) at (6.6, -.6*5) {$(a_4, a_5, \ast)$};
	
	\node[] (b6.list) at (7.9, -.6*6) {$(\diamond, a_6, a_7, a_8, a_9, a_{10}, a_{11}, \diamond)$};
	\node[] (b7.list) at (7.9, -.6*7) {$(\diamond, a_7, a_8, a_9, a_{10}, a_{11}, a_6, \diamond)$};
	\node[] (b8.list) at (7.9, -.6*8) {$(\diamond, a_8, a_9, a_{10}, a_{11}, a_6, a_7, \diamond)$};
	\node[] (b9.list) at (7.9, -.6*9) {$(\diamond, a_9, a_{10}, a_{11}, a_6, a_7, a_8, \diamond)$};
	\node[] (b10.list) at (7.9, -.6*10) {$(\diamond, a_{10}, a_{11}, a_6, a_7, a_8, a_9, \diamond)$};
	\node[] (b11.list) at (7.9, -.6*11) {$(\diamond, a_6, a_7, a_8, a_9, a_{10}, a_{11}, \diamond)$};

\end{tikzpicture}
    \caption{Example of the lower bound construction for finding $B$-optimal matchings. The solid edges
    represent the $ A$-optimal matching $M$ that needs to be shown to be also $ B$-optimal using queries.
    The dashed edges represent rotation edges. 
     Each of the agents in $\{a_0,a_1,\ldots,a_5\}$ also has a rotation edge to some agent in $\{b_6,b_7,b_8,b_9,b_{10},b_{11}\}$ that is not shown.
    An asterisk ($\ast$) indicates that the remaining agents are placed in arbitrary order in the preference list. 
    A diamond ($\diamond$) indicates that the adversary decides in response to the queries made by the algorithm which of the agents in $\{a_0,a_1,\ldots,a_5\}$ are placed at the front of the preference list and which at the back.}
    \label{fig:BoptLB}
    \end{figure*}
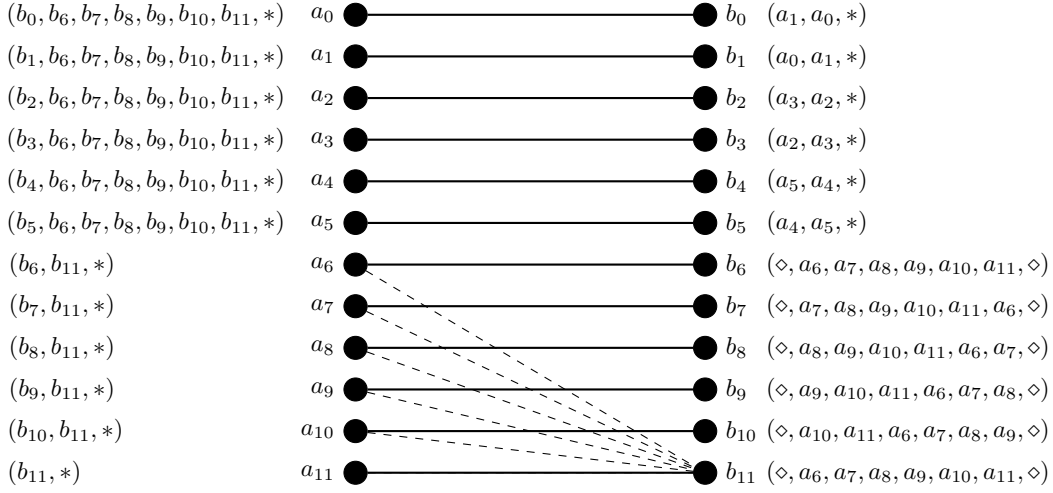

        We partition $A$ into three subsets $A_1 = \{a_0,\ldots,a_{\frac{n}{2}-1}\}$, $A_2 = \{a_{\frac{n}{2}},\ldots,a_{n-2}\}$ and $A_3 = \{a_{n-1}\}$.
        and $B$ into two subsets, $B_1 = \{b_0,\ldots,b_{\frac{n}{2}-1}\}$ and $B_2 = B \setminus B_1$.

        In the following, we
        first define the known A-side preferences and the adversarial strategy. Then we give bounds on the optimal number of queries and the number of queries made by any deterministic algorithm.

        \textbf{A-side preferences.} Consider the following preference lists for $A$.
        For an agent $a_i \in A_1$, the preference list consists of three parts, $P(a_i) = P_1(a_i) P_2(a_i) P_3(a_i)$.
        The first part of the list is the corresponding $i^{th}$ agent of $B$, i.e., $P_1(a_i) = (b_i)$.
        The second part consists of the $\frac{n}{2}$ agents of set $B_2$ in increasing order, i.e., $P_2(a_i) = (b_{\frac{n}{2}}, b_{\frac{n}{2}+1}, \dots, b_{n-2}, b_{n-1})$.
        The last part $P_3(a_i)$ consists of the 
        agents of $B_1\setminus\{b_i\}$ in an arbitrary order.

        For an agent $a_i \in A_2$, the preference list starts with agent $b_i$, followed by the last agent $b_{n-1}$ and finally an arbitrary order of the remaining agents in group $B$.
        For the single agent $a_{n-1}$ in set $A_3$, the preference list starts with agent $b_{n-1}$ followed by an arbitrary order of the remaining agents of set $B$.
          
        \textbf{Adversarial strategy.} The preference lists of the agents in set $B$ are unknown.
        The instance has the $A$-optimal matching $M = \{(a_i,b_i) \mid 0 \le i < n\}$.
        The adversary will ensure that this matching is also $B$-optimal.
        Since each $a_i$ is matched with its top choice, proving stability does not require any queries.
        To prove $B$-optimality of $M$, the executed queries must prove that there is no rotation.
    
        The adversary will ensure that $M$ can be shown to be a $B$-optimal matching
        with $\mathcal{O}(n)$ queries while any deterministic algorithm is forced to make $\Omega(n^2)$ queries.
        
        To achieve this, the adversary sets the preferences of the agents of $ B_1$ independent of the algorithm's actions as follows. For each odd $i \in \{1,3,5,\ldots,(n/2)-1\}$, we let the preference list of $b_i$ start with $a_{i-1}$ followed by $a_i$ and finally all remaining agents in $A$ in an arbitrary order. The preference list of $b_{i-1}$ starts with $a_i$ followed by $a_{i-1}$ and then the remaining agents in $A$ in an arbitrary order.
        Using these preferences, the sequences $(a_{i-1},b_{i-1}),(a_{i},b_{i})$ are potential rotations. To prove that such a sequence is not a rotation, an algorithm has to show that either $(a_{i-1},b_{i})$ or $(a_{i},b_{i-1})$ is not an $r$-edge.
        The only way of showing this is to prove that either $a_{i-1}$ or $a_i$ instead has an $r$-edge to some agent of $B_2$.    
        
        Consider any deterministic algorithm.
        The adversary selects the preferences of the agents in $B_2$ in such a way that the following properties
        hold:
        \begin{description}
        \item[(P1)] Agent $a_{n-1}$ has no $r$-edge. Note that, by the definition of the preferences of $B_1$ above, $a_{n-1}$ already cannot have a rotation edge to an agent of $B_1$.
        \item[(P2)] Each agent in $A_2$ has an $r$-edge to $b_{n-1}$.
        \item[(P3)] Each agent $a_i$ in $A_1$ has an $r$-edge to some agent $b_{t(i)}$ of $B_2$. The choice of that agent $b_{t(i)}$ depends on the queries made by the algorithm.  
        \end{description}
    
        The properties (P1)--(P3) ensure that there is no rotation, as the alternating path starting at any $a \in A \setminus \{a_{n-1}\}$
        with the $r$-edge of that agent ends at~$a_{n-1}$, which has no $r$-edge.
     
        Let $t(i)$ denote the index of the agent $b_{t(i)}$ of $B_2$ to which $a_i\in A_1$ has an $r$-edge.
        This index is determined by the adversary in response to the queries made by the algorithm. 
        Concretely, the
        adversary lets $t(i)$
        be the index of the \emph{last agent} $b_j \in B_2$ for which the algorithm makes a query of the
        form $\prefer(b_j,a_i,\ast)$, where we use $\prefer(b_j,a_i,\ast)$ as a short-hand to refer to queries
        $\prefer(b_j,a_i,a_{i'})$ or $\prefer(b_j,a_{i'},a_i)$ for some~$i'$.
        If the algorithm doesn't
        make queries of this form for all $b_j\in B_2$, then let $t(i)$ be an arbitrary $j$ such that
        the algorithm does not make a query of this form for $b_j\in B_2$. The adversary sets
        the preferences of the agents in $B_2$ in such a way that $b_{t(i)}$ prefers $a_i$ to her partner $a_{t(i)}$ in the A-optimal matching $M$ while all other $b_j \in B_2$ prefer their partner in the A-optimal matching $a_j$  to~$a_i$.
        For example, if the algorithm was to make queries $\prefer(b_j,a_i,a_j)$ for all $b_j\in B_2$ (which it
        might do in order to check whether $a_i$ has an $r$-edge to one of these agents),
        the adversary would answer $\mfalse$ to the first $\frac{n}{2}-1$ such queries and $\mtrue$ to the
        final one.
        
    To achieve the properties (P1)--(P3), the adversary sets the preferences of each agent $b_j$ of $B_2$ as follows:
        \begin{itemize}
        \item $b_j$ prefers $a_i\in A_1$ to $a_j$ if and only if $j=t(i)$.
        \item If $j\neq n-1$, $b_j$ prefers $a_j$ to $a_{j'}$ for all $j'\neq j$, $a_{j'}\in A_2$.
        \item If $j=n-1$, $b_j$ prefers $a_{j'}$ to $a_{j}$ for all $j'\neq j$, $a_{j'}\in A_2$.
        \end{itemize}
        This can be done by letting the preference list of $b_j$ contain first
        the agents $a_i\in A_1$ with $j=t(i)$ in some order,
        then the agents of $A_2$ in some order (only
        ensuring for $b_j$ that $a_j$ comes first among the agents of $A_2$ if
        $j\neq n-1$ and that $a_j$ comes last among the agents of $A_2$ if
        $j=n-1$),
        and finally the agents $a_i\in A_1$ with
        $j\neq t(i)$ in some order.

        \textbf{Upper bound on the optimal query cost.} 
        An optimal solution for the instance can prove that matching $M$ is $B$-optimal by verifying that the properties (P1)--(P3) 
        indeed hold by using at most $n-1+\frac{n}{2}-1+\frac{n}{2}=2n-2$ queries
        as follows:
        \begin{itemize}
        \item The $n-1$ queries $\prefer(b_i,a_{n-1},a_i)=\mfalse$ for $i\le n-2$ show that $a_{n-1}$ has no $r$-edge.
        \item Each of the $\frac{n}{2}-1$ queries $\prefer(b_{n-1},a_{\frac{n}{2}+i},a_{n-1})=\mtrue$ for $0\le i\le \frac{n}{2}-2$ shows
             that $a_{\frac{n}{2}+i}$ has an $r$-edge to $b_{n-1}$. This is because each agent of $A_2$ has $b_{n-1}$ in its preference list directly after its current matching partner. So if $b_{n-1}$ prefers an agent of $A_2$ over its current partner $a_{n-1}$, then this directly gives us an $r$-edge.
        \item Each of the $\frac{n}{2}$ queries $\prefer(b_{t(i)},a_{i},a_{t(i)})=\mtrue$ for $0\le i\le \frac{n}{2}-1$ shows
             that $a_i$ has a rotation edge to some agent in $B_2$. Based on the result of such a query, $a_i$ must have an $r$-edge to either
         $b_{t(i)}$ or to some other agent of $ B_2$ that is higher up in $a_i$'s preference list.
        \end{itemize}
        \textbf{Lower bound on the algorithm's query cost.}
        We provide to the algorithm the information
        that $a_{n-1}$ has no $r$-edge, that each agent of $ A_2$ has an $r$-edge to $b_{n-1}$, and we reveal the full preference lists of all agents in $B_1$.
        Clearly, this extra information can only reduce the number of queries a deterministic algorithm may need as it could simply ignore the~information.
    
        For each agent $a_i \in A_1$, the algorithm will either make queries of the form $\prefer(b_j,a_i,\ast)$
        for all $b_j\in B_2$ or not. Call $a_i$ \emph{resolved} in the former case
        and \emph{unresolved} otherwise. For any resolved agent, the algorithm may have determined that it has an
        $r$-edge to an agent of $B_2$ and hence cannot be part of a rotation. For the unresolved agents,
        the algorithm cannot know whether they have an $r$-edge to an agent in $B_2$.
    
        As argued above, for each odd $i \in \{1,3,5,\ldots,\frac{n}{2}-1\}$, the algorithm has to resolve~either $a_i$ or $a_{i-1}$ to prove that $(a_i,b_i),(a_{i-1},b_{i-1})$ is not a rotation. 
        Thus,
        it must resolve at least $n/4$ agents. 
        For each resolved agent, the algorithm has made
        queries of the form $\prefer(b_j,a_i,\ast)$ for each $b_j\in  B_2$.
        This totals to at least $\frac{n}{4} \cdot \frac{n}{2} \cdot \frac{1}{2}=\frac{n^2}{16} \in \Omega(n^2)$ queries. Note that we divide $\frac{n}{4}\cdot \frac{n}{2}$ by two as
         a single query $\prefer(b_j,a_i,a_{i'})$
        is of the form $\prefer(b_j,a_i,\ast)$ and also
        $\prefer(b_j,a_{i'},\ast)$. 

	Finally, we show how to extend the lower bound to work for randomized algorithms. 
	
	By Yao's principle~\cite{Borodin98,Yao77} we can prove the theorem by giving a randomized instance $\mathcal{R}$ and showing that $$\EX_{R \sim \mathcal{R}}\left[\frac{\ALG(R)}{\OPT(R)}\right] \in \Omega(n)$$ 
	holds for every deterministic algorithm $\ALG$, where $\ALG(R)$ and $\OPT(R)$ denote the number of queries executed by the algorithm and an optimal solution, respectively, for the realization $R$ of the randomized instance $\mathcal{R}$.

	To define the randomized instance $\mathcal{R}$, we take the instance of the deterministic lower bound and introduce randomization into the uncertain preference lists. 
	We define the preference lists of $A$ and $B_1$ without randomization in the same way as before. Similarly, we leave the sub-list defined for
	the agents of $A_2$ in the preference list of an agent of $B_2$ as it is and only randomize the positions of the agents of $A_1$ in the preference lists of $B_2$.
	
	To this end, consider an odd $i \in \{1,3,\ldots,(n/2)-1\}$. 
	The randomized part of the instance uniformly at random picks a tuple $(a_k,b_j)$ with $k \in \{i-1,i\}$ and $b_j \in B_2$. For this selected tuple, we set the preferences such that $b_j$ prefers $a_k$ over its current matching partner $a_j$. For all other tuples $(a_{k'},b_{j'})$ with $k' \in \{i-1,i\}$, $b_{j'} \in B_2$ and either $k \not= k'$ or $j \not= j'$, we set the preference of $b_{j'}$ such that it prefers its current partner over $a_{k'}$.
	
	This can be achieved by letting the preference list of $b_j$ contain first
	the agents $a_k\in A_1$ such that $(a_k,b_j)$ was selected by the randomized procedure above, then the agents of $A_2$ in some non-randomized order (only
	ensuring for $b_j$ that $a_j$ comes first among the agents of $A_2$ if
	$j\neq n-1$ and that $a_j$ comes last among the agents of $A_2$ if
	$j=n-1$), and finally the agents $a_{k'}\in A_1$ such that the tuple $(a_{k'},b_j)$ was \emph{not} selected by the randomized procedure. Ties can be broken according to some arbitrary but fixed order.
 	
 	By defining the preferences in this way, every realized instance still satisfies the properties (P1) and (P2) as defined in the proof of the deterministic lower bound. While the preferences do not satisfy property (P3), they satisfy for each odd $i \in \{1,3,\ldots,(n/2)-1\}$ that either $a_i$ or $a_{i-1}$ has a rotation edge to some agent of $ B_2$. This still implies that, for every realized instance, the matching $M = \{(a_j,b_j) \mid j \in \{0,\ldots,n-1\}\}$ is B-optimal.
 	Slightly adjusting the strategy of the deterministic proof, one can show that $\OPT \le 2n-2$ still holds for each such realization. 
 	This implies 
 	$$\EX_{R \sim \mathcal{R}}\left[\frac{\ALG(R)}{\OPT(R)}\right] \ge \EX_{R \sim \mathcal{R}}\left[\frac{\ALG(R)}{2n-2}\right] =  \frac{\EX_{R \sim\mathcal{R}}[\ALG(R)]}{2n-2}$$ for every deterministic algorithm $\ALG$. So it suffices to show $\EX_{R \sim\mathcal{R}}[\ALG(R)] \in \Omega(n^2)$ to prove the theorem.
 	
 	To that end, consider an arbitrary deterministic algorithm. As argued in the deterministic lower bound proof, the algorithm has to, for each odd $i \in \{1,3,\ldots,(n/2)-1\}$, either prove that $(a_i,b_{i-1})$ or $(a_{i-1},b_{i})$ is not an $r$-edge. By definition of the instance, this requires at least one query of the form $\prefer(b_j,a_k,\ast)$ (as defined in the deterministic lower bound proof) for the tuple $(b_j,a_k)$ with $b_j \in B_2$ and $k \in \{i,i-1\}$ that was drawn by the randomized procedure above for index $i$.
 	The algorithm will have to execute queries of the form $\prefer(b_{j'},a_{k'},\ast)$ with $b_{j'} \in B_2$ and $k' \in \{i,i-1\}$ until it hits a query with $j' = j$ and $k' = k$.
 	We call such a query \emph{successful} if $j' = j$ and $k' = k$ and \emph{unsuccessful} otherwise.
	In the same way, we call the selected tuples \emph{successful} and all other tuples \emph{unsuccessful}.
 	Note that the algorithm might need further queries to prove that either $(a_i,b_{i-1})$ or $(a_{i-1},b_{i})$ is not a rotation edge, but executing at least one successful query is a necessary condition.
 	
 	Consider a fixed odd $i \in \{1,3,\ldots,(n/2)-1\}$. We bound the expected number of queries of the form $\prefer(b_{j'},a_{k'},\ast)$ with $b_{j'} \in B_2$ and $k' \in \{i,i-1\}$ that the algorithm needs until one of them is successful. 	
 	Let $Y_i$ be a random variable denoting the number of queries of that form the algorithm executes.
 	Note that the algorithm might execute different queries in-between the queries of that form, but the random variable $Y_i$ only counts the queries of that form for the fixed $i$ and ignores different queries that are executed in-between them.
 	To further characterize $Y_i$, let $Z_{i,\ell}$ with $\ell \ge 1$ be an indicator random variable denoting whether the first $\ell$ queries of that form are \emph{not} successful. Then,
 	$Y_i = 1 + \sum_{\ell \ge 1} Z_{i,\ell}$ and $\EX_{R \sim \mathcal{R}}[Y_i] = 1 + \sum_{\ell \ge 1} \EX_{R \sim \mathcal{R}}[Z_{i,\ell}]$. 
 	
 	We first observe that queries that do no involve $a_i$ and $a_{i-1}$ do not give any information on which tuples can be successful for $i$. Furthermore, queries that involve $a_{i}$ (or $a_{i-1}$) and some $b_j \in B_2$ do not admit any information on whether some tuple $(a_k,b_{j'})$ (or some tuple $(a_{i-1},b_{j'})$) with $j' \not= j$ or $k=i-1$ (or $k=i$) is successful or not.
 	Thus, at any point during the execution of an algorithm, all tuples $(a_k,b_j)$ for which the algorithm did not yet execute a query of form $\prefer(b_j,a_k, \ast)$ are equally likely to be successful (unless the algorithm already found the successful tuple).
 	
 	Consider the expected value $\EX_{R \sim \mathcal{R}}[Z_{i,\ell}] = \Pr[Z_{i,\ell} = 1]$. For $\ell = 1$, we have $\Pr[Z_{i,\ell} = 1] = \frac{n-2}{n}$ since there are $n$ tuples $(a_{k'}, b_{j'})$ with $k' \in \{i,i-1\}$, among those only one successful tuple is drawn uniformly at random, and a query can cover at most two such tuples at the same time if it is of form $\prefer(b_{j'},a_{i},a_{i-1})$.  
 	For $\ell = 2$, we have $\Pr[Z_{i,\ell} = 1] \ge \frac{n-2}{n} \cdot \frac{n-4}{n-2}$ because given that the first query is not successful there are still $n-2$ tuples that could still be successful, only one uniformly at random selected tuple is actually successful, and the second query can cover at most two of the potentially successful tuples.
 	Continuing this argumentation, we get 
 	$$
 	 \EX_{R \sim \mathcal{R}}[Z_{i,\ell}] = \Pr[Z_{i,\ell} = 1] = \prod_{\ell' =1}^\ell \frac{n-2 \cdot \ell'}{n - 2 \cdot (\ell'-1)} = 1-\frac{2\ell}{n}
 	$$
 	for each $1 \le \ell \le n/2$. This directly implies
 	\begin{align*}
 	\EX_{R \sim \mathcal{R}}[Y_i] & = 1 + \sum_{\ell \ge 1} \EX_{R \sim \mathcal{R}}[Z_{i,\ell}] \ge 1 + \sum_{\ell = 1}^{n/2} \EX_{R \sim \mathcal{R}}[Z_{i,\ell}]\\
 	  & \ge \sum_{\ell = 1}^{n/2} (1 - \frac{2\ell}{n}) = \frac{n-2}{4}.
 	\end{align*}
 	The number of queries the algorithm executes on a realization $R$ is at least  $\ALG(R) \ge \sum_{i \in \{1,3,\ldots,(n/2)-1\}} Y_{i}$, which implies
 	\begin{align*}
 	\EX_{R \sim \mathcal{R}}[\ALG(R)] & \ge\!\! \sum_{i \in \{1,3,\ldots,(n/2)-1\}} \!\! \EX_{R \sim \mathcal{R}}[Y_{i}] \ge \frac{n}{4} \cdot \frac{n-2}{4}\\
 	& = \frac{n^2-2n}{16} \in \Omega(n^2).
 	\end{align*}
    \end{proof}
    
    \subsubsection{Offline Results for Computing
    \texorpdfstring{$B$}{B}-Optimal Stable Matchings}
    \label{subsubsec:pair-B-offline}
    
    We show NP-hardness for the offline problem of verifying a given matching $M$ to be stable and $B$-optimal. Recall that in the offline problem we assume full knowledge of the $B$-side preferences but still want to compute a query set of minimum size that a third party without knowledge of the $B$-side preferences could use to verify the $B$-optimality of $M$.
    
    \begin{theorem}\label{thm:offline-pair-B-NPhard}
    The offline problem of computing an optimal set of \comparison queries for finding (or verifying) the $B$-optimal stable matching in a stable matching instance with one-sided uncertainty is NP-hard.
    \end{theorem}
    
    \begin{proof}
    We give a reduction from the NP-hard {\em Minimum Feedback Arc Set (FAS)} problem.
    Given a directed graph $G=(V,E)$, a feedback arc set is a subset of edges $E'\subseteq E$ which, if removed from $G$, leaves the remaining graph
    acyclic. The FAS problem is to decide  for a given directed graph and some $k\in \mathbb{Z}_+$, whether there is a feedback arc set $E'$
    with $|E'|\leq k$.

    Given an instance of FAS with
    $G=(V,E)$ and some $k$, we construct a stable matching instance with one-sided uncertainty as follows.
    For each node $v$ of $G$, introduce an agent
    $v$ in $ A$ and an agent $v'$ in $ B$.
    Let $N^+(v)$ denote the set of out-neighbors of $v$ in $G$,
    and $d^+(v)=|N^+(v)|$.
    The preference list of $v$ is such that
    it ends with $v'$ followed by all $u'$ for $u\in N^+(v)$.
    All other $w'$ in $ B$ come before $v'$.
    Thus, the elements of $B\setminus\{u' \mid u\in N^+(v)\}$ are the most preferred partners of $v$, followed by $v'$ and finally the elements of $\{u' \mid u\in N^+(v)\}$.
    Let $ M$ be the matching that matches $v$ to $v'$, for all $v$.
    The preference lists of $b\in B$ are such that $ M$ is the
    $ B$-optimal stable matching: Every $v'$ has $v$ as top preference,
    and the remaining agents of $A$ follow in arbitrary order.
    By selecting the matching $M$ this way, we have that, for every $v \in A$, all edges to elements of $\{u' \mid u\in N^+(v)\}$ are potential $r$-edges. To prove that such an edge $(v,u')$ is not an $r$-edge, an algorithm has to compare $u$ and $v$ from the perspective of $u'$ to prove that $u'$ prefers $M(u')=u$ over~$v$.

    The number of queries $Q(M)$ needed to verify the stability of $ M$
    is determined by $M$ and is polynomial-time computable by using~\Cref{thm:1-comp-verification}. 
    To prove $B$-optimality of $M$, we need to show that there is no rotation (\Cref{lem:rotations-for-optimal}). Indeed, there is a query strategy with $k$ queries for verifying that there is no rotation if and only if there is a feedback arc set in~$G$ of size $k$. To see this, observe that every directed cycle in~$G$ corresponds to a potential rotation
    in the matching instance, and every query that excludes one of
    the edges of the potential rotation from being an $r$-edge
    corresponds to the removal of the corresponding arc in $G$.
    
    Note that, for the constructed instance, all queries
    to verify the stability of $M$ obtain information of the form $M(b) \prec_b a$ for $a \in A$ and $b \in B$ with $b \prec_a M(a)$. On the other hand, all queries that help to verify the absence of a rotation obtain information of the form $M(b) \prec_b a$ for $a \in A$ and $b \in B$ with $M(a) \prec_a b$.
    As these are disjoint query sets, we can conclude that there is a query strategy that proves $M$ to be stable and $B$-optimal with at most $Q(M) + k$ queries if and only if there is a feedback arc set in $G$ of size~at~most~$k$.
    \end{proof}
    
    We also prove the following approximation for the offline problem by exploiting an $\mathcal{O}(\log n\log\log n)$-approximation for weighted feedback arc set by Even et al.~\cite{EvenNSS98}.
    
    \begin{theorem}\label{thm:offline-B-pair-approx}
    The offline problem of computing an optimal set of \comparison queries for finding the $B$-optimal stable matching in a stable matching instance with one-sided uncertainty can be approximated within ratio $\mathcal{O}(\log n\log\log n)$.
    \end{theorem}
    
\begin{proof}
	Let $ M$ be the $ B$-optimal matching. We give an algorithm that verifies $M$ to be stable and $B$-optimal by executing at most $\mathcal{O}(\log n\log\log n) \cdot \OPT$ queries, where $\OPT$ is the optimal number of queries for the same instance.
	First, the algorithm proves that $M$ is stable using~\Cref{thm:1-comp-verification}. This leads to at most $\OPT$ queries.
	
	After that, the algorithm has to prove $B$-optimality. First, for every $a \in A$ that has an $r$-edge to an agent $r(a) \in B $, the algorithm queries $\prefer(r(a), a, M(r(a)))$. Since $(a,r(a))$ is an $r$-edge, this query must return that $r(a)$ prefers $a$ over $M(r(a))$. This leads to at most $n \le \OPT+1$ queries ($n \le \OPT+1$ holds by~\Cref{lem:bopt:numberlb:pw}).
	Note that, for an $a \in A$ with an $r$-edge, the query $\prefer(r(a), a, M(r(a)))$ proves that $a$ has an $r$-edge but is not necessarily sufficient to prove that $(a,r(a))$ is indeed the $r$-edge of $a$. If there is an agent $b \in B$  with $M(a) \prec_a b \prec_a r(a)$ for which we have not yet verified whether $b$ prefers $a$ over $M(b)$, then $(a,b)$ could also still be the $r$-edge of $a$. We call such pairs $(a,b)$ \emph{potential $r$-edges} and let $P$ denote the set of these edges.
	
	It remains to consider the graph $G$ defined by the matching edges, the $r$-edges $R$, and all potential $r$-edges $P$. If $G$ has no cycle alternating between edges in $M$ and edges in $P \cup R$, then we have shown that $M$ does not expose a rotation and, thus, is $B$-optimal.  Otherwise, the algorithm has to execute queries $\prefer(b,a,M(b))$ for edges $(a,b) \in P$ to prove that they are not actually $r$-edges until it becomes clear that $M$ has no rotation.
	
	To select the edges $(a,b) \in P$ for which the algorithm executes such queries, we exploit the $\mathcal{O}(\log n\log\log n)$-approximation for weighted feedback arc set by Even et al.~\cite{EvenNSS98}.
	To this end, we create an instance of the weighted feedback arc set problem by considering the vertices $A \cup B$, adding the edges $M \cup R$ with weight $\infty$ each and adding the edges $P$ with weight $1$ each. We orient all edges in $M$ from the $B$-side vertex to the $A$-side vertex and all edges in $R \cup P$ from the $A$-side vertex to the $B$-side vertex. The orientation ensures that all cycles in the graph alternate between $M$-edges and $R \cup P$-edges. Since the matching $M$ is $B$-optimal by assumption, there cannot be an alternating cycle using only edges in $M \cup R$, so there must be a feedback arc set that only uses edges in $P$. The choice of the edge weights ensures that every approximation algorithm for weighted feedback arc set finds such a solution.
	We use the $\mathcal{O}(\log n\log\log n)$-approximation to find such a feedback arc set $F \subseteq P$. Since removing $F$ from the instance yields an acyclic graph, querying $\prefer(b,a, M(b))$ for each $(a,b) \in F$ proves that $M$ does not expose a rotation. As the minimum weight feedback arc set is the cheapest way to prove that $M$ does not have a rotation, we have $|F| \le \mathcal{O}(\log n\log\log n) \cdot \OPT$, which implies the theorem.
\end{proof}
    
\subsection{Verifying a Stable Matching with Two-Sided Uncertainty}
We observe that the lower bound on the optimal number of queries in \Cref{cor:relinst} can also be used for verifying a stable matching in a stable matching instance with uncertain preferences on both sides.

\begin{theorem}\label{thm:2-comp-verification}
    In the \comparison query model, there is a $2$-competitive algorithm for verifying that a given matching $M$ in a stable matching instance with uncertain preferences on both sides is stable.
\end{theorem}

\begin{proof}
	To verify that a given matching $M$ in a graph $G$ is stable, we have to prove for each $(a,b) \notin M$ that $(a,b)$ is not a blocking pair. That is, we have to prove that $M(a) \prec_a b$ or that $M(b) \prec_b a$. Note that $M(a) \prec_a b$ (or $M(b) \prec_b a$) can be verified by directly querying $\prefer(a,b,M(a))$ (or $\prefer(b,a,M(b))$) or indirectly via transitivity.

	For every pair $(a,b)\notin M$, let the algorithm query $\prefer(b,a,M(b))$
	first and, if the answer is that $a \prec_b M(b)$, query also $\prefer(a,b,M(a))$.
	If the answer to the latter query is that $b\prec_a M(a)$, the pair $(a,b)$ is
	a blocking pair, and the algorithm outputs that $M$ is not a stable matching.
	If the algorithm finds for every pair $(a,b)\notin M$ that
	$M(b) \prec_b a$ or $M(a)\prec_a b$, the algorithm outputs that $M$ is a stable
	matching.
  
	The algorithm makes at most $2(n^2-n)$ queries, as it makes at most $2$ queries
	for each of the $n^2-n$ pairs $(a,b)\notin M$.

	We now show that the optimal number of queries is at least
	$n^2-n$. For each pair $(a,b)\notin M$, the optimum needs to prove
	$M(a) \prec_a b$ or $M(b)\prec_b a$. This means it must relate
	$b$ to $M(a)$ for $a$, or it must relate $a$ to $M(b)$ for~$b$.
	Either way, this produces a relationship pair. As no two
	different pairs $(a,b)\notin M$ can produce the same relationship
	pair, the total number of relationship pairs is at least
	$n^2-n$. By Corollary~\ref{cor:relinst}, this implies that the optimum makes
	at least $n^2-n$ queries. As the algorithm makes at most
	$2(n^2-n)$ queries, it is $2$-competitive.
\end{proof}

We can show that no deterministic algorithm can do better.

\begin{theorem}\label{thm:lb-verification}
    In the \comparison query model, no deterministic algorithm can be better than $2$-competitive for the problem of verifying that a given matching $M$ in a stable matching instance with uncertain preferences on both sides is stable.
\end{theorem}

\begin{proof}
    Let $A=\{a_1,a_2\}$, $B= \{b_1,b_2\}$ and $M = \{(a_1,b_1),(a_2,b_2)\}$. Any algorithm has to prove that $(a_1,b_2)$ and $(a_2,b_1)$ are not blocking pairs.
    Thus, any algorithm has to either prove $a_2 \prec_{b_2} a_1$ or $b_1 \prec_{a_1} b_2$ (to verify that $(a_1,b_2)$ is not a blocking pair) and $a_1 \prec_{b_1} a_2$ or $b_2 \prec_{a_2} b_1$ (to verify that $(a_2,b_1)$ is not a blocking pair).

    Since the subproblems of proving that $(a_1,b_2)$ and $(a_2,b_1)$ are not blocking pairs are independent of each other, we can w.l.o.g.~assume that the algorithm starts by proving that $(a_1,b_2)$ is not a blocking pair.
    If the algorithm starts by querying $\prefer(b_2,a_1,a_2)$, then the adversary reveals $a_2 \succ_{b_2} a_1$, which forces the algorithm to also query $\prefer(a_1,b_1,b_2)$. We let this query reveal $b_1 \prec_{a_1} b_2$. The optimal solution only queries  $\prefer(a_1,b_1,b_2)$. If the algorithm starts by querying  $\prefer(a_1,b_1,b_2)$, we can argue symmetrically. Thus, the algorithm executes twice as many queries as the optimal solution to prove that $(a_1,b_2)$ is not a blocking pair.

    We can argue analogously to show that the algorithm also executes twice as many queries as the optimal solution to prove that $(a_2,b_1)$ is not a blocking pair, which implies the result.
\end{proof}

    \section{Stable Matching with Interview Queries}
    \label{sec:interview}
    
    In this section, we consider the interview query model.
    Most of our results and proofs are quite similar to their counterparts for comparison queries. This might be surprising as interview and comparison queries are, in a sense, incomparable: While interview queries allow us to more efficiently determine full preference lists, a comparison between two agents can be done more efficiently via a single comparison query.
    As we show the same (asymptotic) bounds on the competitive ratio, the latter seems to be the deciding factor.

\subsection{Verifying and Finding a Stable Matching with Interview Queries}
A $1$-competitive algorithm for finding a stable matching and verifying a given stable matching with interview queries is implied by the results and arguments from \cite{DBLP:conf/sigecom/RastegariCIL13} for a more general uncertainty setting and can be derived as follows.
    
Consider a given instance of stable matching with one-sided uncertainty and a given stable matching $M$. To verify that $M$ is indeed stable, we have to consider all potential blocking pairs, i.e., all pairs $(a,b)$ with $b \prec_a M(a)$. For such a pair, we have to verify that $M(b) \prec_b a$ holds to prove that $(a,b)$ is not a blocking pair. The only way of comparing $M(b)$ and $a$ from $b$'s perspective is to execute the interviews $\intq(b,a)$ \emph{and} $\intq(b,M(b))$.
For a fixed $b \in B$, this implies that the minimum number of interviews involving $b$ necessary to prove that $M$ is stable is $Q_b(M) = 0$ if no element of $a \in A \setminus \{M(b)\}$ prefers $b$ over its current partner $M(a)$ and $Q_b(M) = 1 + |\{ a \in A \mid b \prec_a M(a)\}|$ otherwise. We can observe the following.

\begin{observation}
Consider a given instance of stable matching with one-sided uncertainty and a given stable matching $M$. The minimum number of interview queries necessary to verify that $M$ is indeed stable is $Q(M) = \sum_{b \in B} Q_b(M)$ with $Q_b(M) = 0$ if no element of $a \in A \setminus \{M(b)\}$ prefers $b$ over its current partner $M(a)$ and $Q_b(M) = 1 + |\{ a \in A \mid b \prec_a M(a)\}|$ otherwise 
\end{observation}

Consider the following algorithm: For each $b \in B$ with $Q_b(M) > 0$, query $\intq(b,M(b))$ and $\intq(b,a)$ for each $a \in A$ with $b \prec_a M(a)$. This algorithms algorithm clearly verifies the stability of $M$ and executes exactly $\sum_{b \in B} Q_b(M)$ interview queries. Thus, the observation implies the following lemma.

\begin{lemma}
    For a given stable matching instance with one-sided uncertainty and a stable matching~$M$, there is a $1$-competitive algorithm for verifying that $M$ is stable in the interview query model.
\end{lemma}

Similar to the \comparison query model, we can observe that the $A$-optimal matching $M^*$ minimizes the query cost for verifying stability $Q(M) = \sum_{b \in B} Q_b(M)$ over all stable matchings $M$.

To find such an $A$-optimal matching, we can again just consider the deferred acceptance algorithm where $A$
makes the proposals. Whenever an agent $a\in A$ makes a proposal to an element $b\in B$ that is currently matched to some $a'\in A$, the algorithm queries $\intq(b,a)$ and $\intq(b,a')$. Each interview query is only executed if it has not yet been queried during the previous execution if the algorithm. If the query result is that $b$ prefers $a$ to $a'$, then $b$ accepts $a$'s proposal and becomes matched to $a$ while
$a'$ becomes unmatched. Otherwise, $b$ rejects the proposal and remains matched to~$a'$.

It is not hard to see that this algorithm executes exactly $Q(M^*) = \sum_{b \in B} Q_b(M^*)$ interview queries. This implies that the deferred acceptance algorithm is $1$-competitive for finding the $A$-optimal matching or any stable matching with interview queries.

\subsection{Finding a \texorpdfstring{$B$}{B}-Optimal Stable Matching with Interview Queries}

For finding a $B$-optimal stable matching with interview queries, it is not hard to see that the lower bound of~\Cref{thm:B-opt:det-lb} for \comparison queries nearly directly translates. We briefly sketch how to adjust that lower bound for interview queries to achieve the following theorem.

\begin{theorem}
	\label{thm:int:LB}
	In the interview query model, every deterministic or randomized online algorithm for finding a $B$-optimal stable matching in a stable matching instance with one-sided uncertainty has competitive ratio~$\Omega(n)$.
\end{theorem}

\begin{proof}[Proof sketch.]
	
	We separately sketch the deterministic and randomized lower bound.
	
	\textbf{Deterministic lower bound.} Consider the same instance as in the deterministic lower bound of~\Cref{thm:B-opt:det-lb}. Recall that each $a_i \in A_1$ has an $r$-edge to some $t(i) \in B_2$ that is selected by the adversary depending on the queries executed by the deterministic algorithm. Fix an $a_i \in A_1$. For interview queries we select $t(i)$ as the last element $b \in B_2$ for which the algorithm executes a query $\intq(b,a_i)$. If the algorithm does not execute such a query for every element of $B_2$, then select an arbitrary agent $b$ of $B_2$ for which the query $\intq(b,a_i)$ has \emph{not} been executed by the algorithm.
	
	For a fixed $a_i \in A_1$, this forces any deterministic algorithm to execute at least $|B_2|$ queries $\intq(b,a_i)$ with $b \in B_2$ to prove that $a_i$ has an $r$-edge to some $b \in B_2$. As argued in the proof for comparison queries, any deterministic algorithm has to do this for at least $\frac{n}{4}$ members of $A_1$. This leads to a total of at least $\frac{n}{2}\cdot \frac{n}{4} \in \Omega(n^2)$ interview queries for every deterministic algorithm. 
	
	The optimal solution on the other hand needs at most $\mathcal{O}(n)$ queries by for example executing the query strategy described in the proof for \comparison queries while simulating each comparison query with at most two interviews.
	
	\textbf{Randomized lower bound.} For the randomized lower bound, we again use Yao's principle, consider the same randomized instance as in the proof for \comparison queries and prove that every deterministic algorithm need $\Omega(n^2)$ queries in expectation. 
	
	To this end, consider an arbitrary deterministic algorithm.
	Recall that, for each odd $i \in \{1,3,\ldots,(n/2)-1\}$, the algorithm either has to prove that $(a_i,b_{i-1})$ or $(a_{i-1},b_{i})$ is not an $r$-edge. By definition of the instance, this requires at least one query of the form $\intq(b_j,a_k)$ for the tuple $(b_j,a_k)$ with $b_j \in B_2$ and $k \in \{i,i-1\}$ that was drawn by the randomized procedure as defined in the \comparison query proof for index $i$.
	The algorithm will have to execute queries of the form $\intq(b_{j'},a_{k'})$ with $b_{j'} \in B_2$ and $k' \in \{i,i-1\}$ until it hits a query with $j' = j$ and $k' = k$.
	We call such a query \emph{successful} if $j' = j$ and $k' = k$ and \emph{unsuccessful} otherwise.
	In the same way, we call the selected tuples \emph{successful} and all other tuples \emph{unsuccessful}.
	Note that the algorithm might need further queries to prove that either $(a_i,b_{i-1})$ or $(a_{i-1},b_{i})$ is not a rotation edge, but executing at least one successful query is a necessary condition. 
	
	After this slight adjustment, we can bound the expected number of queries in the same way as before (with the only difference that a single query can now cover only a single tuple and not two) to prove that every deterministic algorithm makes $\Omega(n^2)$ queries in expectation. On the other hand, the optimal solution for each realization of the randomized instance needs at most $\mathcal{O}(n)$ queries by again simulating the comparison query strategy with at most two interview queries per comparison.
\end{proof}

For the matching upper bound, recall that $n^2$ interview queries are enough to determine the full $B$-side preference lists. This means that we need at most $n^2$ interview queries to find the $B$-optimal matching $M$. We show that even the optimal solution needs at least $\Omega(n)$ interviews to find the $B$-optimal matching, which then implies an $\mathcal{O}(n)$-competitive algorithm. We prove the following lemma by essentially repeating the corresponding proof for \comparison queries (cf.~\Cref{lem:bopt:numberlb:pw})

\begin{lemma}
	The optimal number of queries for verifying 
	the $B$-optimal stable matching with interview queries is at least $n-1$ for every instance of the stable matching problem with one-sided uncertainty.
\end{lemma}

\begin{proof}
Let $M$ be the $B$-optimal stable matching for the given instance. Consider an arbitrary algorithm that verifies $M$ to be $B$-optimal with interview queries.
Assume that there are at least two distinct members $a$ and $a'$ of $B$ for which the algorithm does not execute any queries.
If some $b \in B$ satisfies either $b \prec_a M(a)$ or $b \prec_{a'} M(a')$, then this is a contradiction to the algorithm verifying $M$ to be stable. 
Otherwise, $(a,M(a)),(a',M(a'))$ is a potential rotation so the algorithm has to prove that either $(a,M(a'))$ or $(a',M(a))$ is not an $r$-edge. Assume w.l.o.g.~that the algorithm proves $(a,M(a'))$ to not be an $r$-edge. To do this, it either has to prove $a' \prec_{M(a')} a$, which is impossible without executing the interview $\intq(M(a'),a)$, or it has to prove $a \prec_b M(b)$ for some $b$ with $M(a) \prec_a b \prec_a M(a')$, which is impossible without executing the interview $\intq(b,a)$. Each case leads to a contradiction.
\end{proof}

Together with~\Cref{thm:int:LB}, this lemma implies the following theorem.

    \begin{theorem}
    \label{thm:int-Bopt-Thetan}
    In the interview query model, the best possible (randomized) competitive ratio for finding the $B$-optimal stable matching in an instance of stable matching with one-sided uncertainty is in $\Theta(n)$.
    \end{theorem}

\subsection{NP-Hardness of the Offline Problem}

    For the offline problem of verifying a given $B$-optimal stable matching with interview queries, Rastegari et al.~\cite{DBLP:conf/sigecom/RastegariCIL13} show NP-hardness in a setting
    with partial uncertainty on both sides. As their proof exploits the possibility of giving partial information as part of the input, it does not directly translate to our setting with one-sided uncertainty. 
    However, we can show with a similar proof as for comparison queries that the problem remains hard even in our setting.
    
    \begin{theorem}
    \label{thm:int-offline}
    The offline problem of computing an optimal set of interview queries for finding the $B$-optimal stable matching in a stable matching instance with one-sided uncertainty is NP-hard.
    \end{theorem}

\begin{proof}
	Consider the same construction as in the proof for comparison queries (cf.~\Cref{thm:offline-pair-B-NPhard}). We add a dummy element $z$ to $A$ and a dummy element $z'$ to $B$. Element $z'$ has $z$ as the top choice and afterwards all other agents of $A$ in an arbitrary order. Agent $z$ has $z'$ as the last choice and before that the other elements of $B$ in an arbitrary order. The agents of $A\setminus \{z\}$ all have $z'$ as the top choice and afterwards the preference list as defined in the proof of~\Cref{thm:offline-pair-B-NPhard}.  The elements of $B \setminus \{z'\}$ have $z$ as the last choice and before that the preference list as defined in the proof of~\Cref{thm:offline-pair-B-NPhard}. This forces $(z,z')$ to be part of the $B$-optimal matching and any algorithm has to query $\intq(b,z)$ and $\intq(b,M(b))$  for all $b \in B\setminus\{z'\}$ to prove stability. 
	
	After proving stability, each query $\intq(b,a)$ for an $a \in A$ and $b \in B$ contains the information of $\prefer(b,a,M(b))$ (as $\intq(b,M(b))$ has already been queried to prove stability). Thus, we can now repeat the remaining part of the proof of~\Cref{thm:offline-pair-B-NPhard} to show the theorem.	
\end{proof}

    \section{Stable Matching with Set Queries}
    \label{sec:set}
    
We consider the stable matching problem with one-sided uncertainty and set queries.
    Note that set queries are a natural generalization of \comparison queries.
    For verifying any $B$-optimal matching, we show that the optimal number of set queries is at
    least $n-1$. We also observe that there is an algorithm that makes at most
    $n^2$ queries for finding the $B$-optimal matching (or an $A$-optimal
    matching if we want to), as one can sort all preference
    lists using $n^2$ set queries. This implies an $\mathcal{O}(n)$-competitive algorithm for finding the $B$-optimal matching.
    For the subproblem of verifying that a given matching is $B$-optimal, we give an $\mathcal{O}(\log n)$-competitive algorithm by exploiting the additional power of set queries in an involved binary search algorithm.
    If we only have to verify stability for a given matching, we give a $1$-competitive algorithm.
    Furthermore, we show that the offline problem of verifying that a given matching does not have a rotation is NP-hard.

\subsection{Verifying That a Given Matching Is Stable}

We start by characterizing the optimal number of queries (and query strategy) to verify that a given matching $M$ is stable. The main difference to the \comparison model is that, for a fixed $b \in B$, a single query $\topq(b, \{a \mid b \prec_a M(a)\} \cup \{M(b)\})$ is sufficient to prove that $b$ is not part of any blocking pair.

\begin{theorem}
	\label{thm:setqueries:stable}
	Consider a stable matching instance with one-sided uncertainty and a stable matching $M$. The minimum number of set queries to verify that $M$ is stable is $|\{b \in B \mid \exists a \in A \colon b \prec_a M(a)\}|\le n$. Further, there is a $1$-competitive algorithm to verify that $M$ is stable.
\end{theorem}

\begin{proof}
	Consider an arbitrary $b\in B$. Let $Z(b)=\{a \in A \mid b \prec_a M(a)\}$, i.e., $Z(b)$ contains all $a \in A$ that could potentially form a blocking pair with $b$. Thus, $M$ can only be stable if $M(b) \prec_b a$ holds for all $b \in B$ and $a \in Z(b)$. 
	If $Z(b)\neq \emptyset$, at least one query to $b$ is necessary, and
	the query $\topq(b,Z(b)\cup\{M(b)\})$ with answer $M(b)$ reveals all the required information to prove that $b$ is not part of any blocking pair.
	Thus, the minimum number of queries to confirm that $M$ is stable is $|\{b \in B \mid \exists a \in A \colon b \prec_a M(a)\}|$ as claimed. Furthermore, the algorithm that queries $\topq(b,Z(b)\cup\{M(b)\})$ for all $b \in B$ with $Z(b)\not= \emptyset$ is $1$-competitive.
\end{proof}

\subsection{Verifying That a Given Matching Is Stable and \texorpdfstring{$B$}{B}-Optimal}

For the problem of confirming that a given matching is $B$-optimal by using set queries, we show that every algorithm needs to execute at least $n-1$ queries. This is analogous to the setting with \comparison queries and uses a similar proof as~\Cref{lem:bopt:numberlb:pw}. 
It implies that finding a $B$-optimal matching also requires at least $n-1$ queries.

\begin{lemma}
	\label{lem:bopt:numberlb:sq}
	Consider an arbitrary stable matching instance with one-sided uncertainty and the	
	$B$-optimal matching $M$. Every algorithm needs at least $n-1$ set queries to verify that~$M$ is indeed stable and $B$-optimal.
\end{lemma}

\begin{proof}
	For each $b \in B$, let $
	Z(b) = \{a \in A \mid b \prec_{a} M(a)\}$ and let $S = \{b \in B \mid Z(b)\neq\emptyset\}$.
	By the proof of~\Cref{thm:setqueries:stable}, every algorithm needs to execute at least one query of the form $top(b,X)$ with $X \subseteq A$ for all $b \in S$ and this query has to return $M(b)$ as the top choice. 
	Since verifying $B$-optimality includes proving stability, this leads to at least $|S|$ queries.
	
	Consider an arbitrary algorithm that verifies $M$ to be $B$-optimal and let $A_1 \subseteq A$ denote the agents of $A$ that are returned as the top choice by some query of the algorithm. Then $|S| \le |A_1|$ and $\{ a \in A \mid \exists b \in S \colon M(b) = a\} \subseteq A_1$ by the argumentation above.

	If $|A_1| \ge n-1$, then the statement follows immediately, so assume $|A_1| < n-1$ and let $A_2 = A \setminus A_1$.
	Since $|A_1| < n-1$, the set $A_2$ has at least two distinct members $a_1$ and $a_2$. 
	Furthermore, we must have $M(a_1),M(a_2) \notin S$ as observed above. 
	By definition of $S$, we have $M(a_1) \prec_{a_1} M(a_2)$ and $M(a_2) \prec_{a_2} M(a_1)$. This means that $(a_1,M(a_2))$ and $(a_2,M(a_1))$, based on the initially given information, could potentially be rotation edges. Thus, $(a_1,M(a_1)), (a_2,M(a_2))$ could potentially be a rotation and the algorithm has to prove that this is not the case by showing that one of $(a_1,M(a_2))$ and $(a_2,M(a_1))$ is not an $r$-edge.
	To prove that $(a_1,M(a_2))$ is not an $r$-edge, one has to either verify $a_2 \prec_{M(a_2)} a_1$ or $a_1 \prec_b M(b)$ for some $b \in B$ with $M(a_1) \prec_{a_1} b \prec_{a_1} M(a_2)$. However, this requires at least one query that returns either $a_1$ or $a_2$ as the top choice, and there is a symmetric argument for proving that $(a_2,M(a_1))$ is not an $r$-edge.
	Since $a_1$ and $a_2$ are never returned as the top choice by a query of the algorithm, this is a contradiction to the assumption that the algorithm verifies that $M$ is $B$-optimal.
\end{proof}

In contrast to the \comparison model, there exists an offline algorithm that asymptotically matches the lower bound of~\Cref{lem:bopt:numberlb:sq}. 

\begin{theorem}
	\label{thm:sqOfflineApprox}
	There exists a polynomial-time offline algorithm that, given an instance of stable matching with one-sided uncertainty and the $B$-optimal matching $M$, verifies that $M$ is indeed stable and $B$-optimal by executing $\mathcal{O}(n)$ set queries.
\end{theorem}

\begin{proof}
	By the proof of~\Cref{thm:setqueries:stable}, an algorithm can prove $M$ to be stable by executing at most $n$ set queries, so it remains to prove that $M$ is $B$-optimal by executing at most $\mathcal{O}(n)$ set queries.
	
	We do so by proving that $M$ does not contain a rotation. First, for each $b \in B$, we compute the set $P(b) = \{ a \in A \mid M(a) \prec_a b \text{ and } M(b) \prec_b a \}$. Each tuple $(b,a)$ with $b\in B$ and $a \in P(b)$ could be a rotation edge based on $\prec_a$ but is not a rotation edge as $M(b) \prec_b a$. An algorithm can prove that none of these edge are actually rotation edges by executing a query $\topq(b,P(b)\cup\{M(b)\})$ for each $b \in B$. This leads to $n$ additional queries.
	
	If an $a \in A$ does not have a rotation edge, then the previous queries prove that this is the case. Consider an $a \in A$ that has a rotation edge. Then the second endpoint of that edge is the agent $b \in B$ of highest preference according to $\prec_a$ among those agents that satisfy $M(a) \prec_a b$ and $a \prec_b M(b)$. Let $b$ be that endpoint. To prove that $(a,b)$ is indeed a rotation edge, an algorithm has to verify $a \prec_b M(b)$ and $M(b') \prec_{b'} a$ for all $b'$ with $M(a) \prec_a b' \prec_a b$. The latter has already been verified by the previous $n$ queries and the former can be proven by an additional query $\topq(b, \{a,M(b)\})$. Doing this for every $a \in A$ that has a rotation edge leads to at most $n$ further queries.
	
	Executing these queries yields, for each $a \in A$, either the rotation edge of $a$ or a proof that $a$ does not have a rotation edge. Thus, it gives sufficient information to show that $M$ does not have a rotation and is $B$-optimal. 
\end{proof}

Next, we give an online algorithm that decides whether a given matching $M$ is $B$-optimal by executing at most $\mathcal{O}(n \log n)$ set queries. In combination with \Cref{lem:bopt:numberlb:sq}, this yields an $\mathcal{O}(\log n)$-competitive algorithm for verifying that a given matching is $B$-optimal with set queries. 

\begin{theorem}
	\label{thm:set-alg-Bopt}
	There is an algorithm that decides if a given matching $M$ in a stable matching instance with one-sided uncertainty is stable and $B$-optimal with $\mathcal{O}(n \log n)$~set~queries.
\end{theorem}

\begin{proof}
	First, we can use~\Cref{thm:setqueries:stable} and execute $\mathcal{O}(n)$ queries to decide whether $M$ is stable. If $M$ turns out not to be stable, then we are done.
	Otherwise, we have to decide whether $M$ is $B$-optimal by using at most $\mathcal{O}(n \log n)$ set queries. We do so by giving an algorithm that, for each $a \in A$, either finds the rotation edge of $a$ or proves that $a$ does not have a rotation edge. After executing that algorithm we clearly have sufficient information to decide whether $M$ exposes a rotation and, thus, whether it is $B$-optimal.

	For each $a \in A$, we use $R(a)$ to refer to the set of agents that could potentially form a rotation edge with $a$. Initially, we set $R(a) = \{ b \in B \mid M(a) \prec_a b\}$ as all agents with a lower priority than $M(a)$ can potentially form a rotation edge with $a$ based on the initially given information.
	During the course of our algorithm, we will update the set $R(a)$ such that it always only contains the agents of $B$ that, based on the information obtained by all previous queries, could still form a rotation edge with $a$.
	In particular, if we obtain the information that $M(b) \prec_b a$ for some $b \in R(a)$, then $(a,b)$ clearly cannot be a rotation edge and we can update $R(a) = R(a) \setminus \{b\}$. Similarly, if we obtain the information that $a \prec_b M(b)$ for some $b \in R(a)$, then the agents $b' \in R(a)$ with $b \prec_a b'$ cannot form a rotation edge with $a$ anymore and we can update $R(a) = R(a) \setminus \{b' \in R(a) \mid b \prec_a b'\}$.
	Given the current list $R(a)$ of potential rotation edge partners, we use $\bar{R}(a)$ to refer to the $\left\lceil\frac{|R(a)|}{2}\right\rceil$ agents of $R(a)$ with the highest priority in $R(a)$ according to $\prec_{a}$.
	
	\begin{algorithm}[t]
		\KwIn{Stable matching instance with one-sided uncertainty and a matching $M$.}
		Decide whether $M$ is stable using~\Cref{thm:setqueries:stable}. If $M$ is not stable, terminate\;
		$R(a) \gets \{ b \in B \mid M(a) \prec_a b\}$ for all $a \in A$\;
		\While{We did not decide yet whether $M$ is $B$-optimal}{
			$U \gets \{a \in A \mid |\bar{R}(a)| \ge 1\}$\label{alg:setqueries:verify:start}\;
			\For{$b \in B$}{
				$U_b \gets \{a \in U \mid b \in \bar{R}(a)\}$\;
				\Repeat{$U_b = \emptyset$\label{alg:setqueries:verify:end}}{
					$t \gets \topq(b, U_b \cup \{M(b)\})$\;
					\lIf{$t = M(b)$}{
						$R(a) \gets R(a) \setminus b$ for all $a \in U_b$;
						$U_b \gets \emptyset$;
					}\Else{
						$U \gets U \setminus \{t\}$;
						$U_b \gets U_b \setminus \{t\}$\;
						$R(t) \gets R(t) \setminus \{ b' \in R(t) \mid b \prec_t b'\}$\;
					}
				}	
				
			}

		}
		\caption{Algorithm to decide whether a given matching is $B$-optimal using set queries.}
		\label{alg:setqueries:verify}
	\end{algorithm} 
	
	Our algorithm, cf.~Algorithm~\ref{alg:setqueries:verify}, proceeds in iterations that each execute at most $\mathcal{O}(n)$ set queries. Let $R_i(a)$, $a \in A$, denote the current sets of potential rotation edges at the beginning of iteration $i$ and let $\bar{R}_i(a)$ be as defined above.
	We define our algorithm in a way such that each iteration $i$ decides for each $a \in A$ whether it has a rotation edge to an agent of $\bar{R}_i(a)$ or not. 
	Then, $|R_{i+1}(a)| \le \frac{|R_i(a)|+1}{2}$ holds for each $a \in A$ with $|R_i(a)| > 1$ as we either get $R_{i+1}(a) \subseteq \bar{R}_i(a)$ or $R_{i+1} \subseteq R_i(a) \setminus \bar{R}_i(a)$.  
	Furthermore, if $|R_i(a)| = 1$, then iteration $i$
	either identifies the rotation edge of $a$ or proves that it does not have one.
	This means that after at most $\mathcal{O}(\log n)$ such iterations, for each $a \in A$, we either found the rotation edge of $a$ or verified that it does not have one.
	Since each iteration executes $\mathcal{O}(n)$ set queries, we get an algorithm that executes $\mathcal{O}(n \log n)$ set queries and decides whether $M$ is $B$-optimal.
	
	It remains to show that each iteration $i$ indeed executes $\mathcal{O}(n)$ set queries and decides, for each $a \in A$, whether $a$ has a rotation edge to some agent of $\bar{R}_i(a)$. Lines~\ref{alg:setqueries:verify:start} to~\ref{alg:setqueries:verify:end} of Algorithm~\ref{alg:setqueries:verify} show the pseudocode for such an iteration.
	In each iteration $i$, the algorithm considers the set $U = \{a \in A \mid |R_i(a)| \ge 1\}$, i.e., the subset of $A$ for which we do not yet know whether it has a rotation edge to some agent of $\bar{R}_i(a)$. Then, the algorithm iterates through the agents $b$ of $B$ and considers the set $U_b = \{a \in U \mid b \in \bar{R}_i(a)\}$. Note that, for each $a \in U_b$, it holds that if $a \prec_b M(b)$, then $a$ has a rotation edge to some agent of $\bar{R}_i(a)$ (not necessarily to $b$).
	The algorithm
	executes the query $\topq(b,U_b \cup \{M(b)\})$. If this query returns $M(b)$, then we know for sure that $b$ does not have a rotation edge to any agent of $U_b$ and we can discard $b$ for the rest of the iteration and also remove $b$ from the current $R(a)$ of all $a \in U_b$. On the other hand, if the query returns $a \not= M(b)$, then we know that $a$ has a rotation edge to some agent of $\bar{R}_i(a)$ and we do not need to consider $a$ for the rest of the iteration anymore. Thus, after each query within the iteration we discard an agent of either $A$ or $B$, which means that the iteration terminates after at most $2n$ queries. At the end of the iteration, we know for each $a \in A$ whether it has a rotation edge to some $b \in \bar{R}_i(a)$.
\end{proof}

For the offline problem, we show that computing the query set of minimum size that verifies that a given matching does not have a rotation is NP-hard. However, in the instances constructed by the reduction, verifying that the given matching does not have a rotation \emph{and is stable} is trivial as we will discuss after the proof. This means that the following result does \emph{not} imply NP-hardness for the offline variant of finding the $B$-optimal matching with set queries.

\begin{theorem}
	In the set query model, the offline problem of computing an optimal set of queries for
	verifying that a given $B$-optimal stable matching $M$ for a stable matching instance with one-sided uncertainty
	does not have a rotation is NP-hard.
\end{theorem}

\begin{proof}
	We show the statement by reduction from the NP-hard \emph{feedback vertex set problem}~\cite{Karp10}. In this problem, we are given a directed graph $G=(V,E)$ and a parameter $k \in \mathbb{N}$. The goal is to decide whether there exists a subset $F \subseteq V$ with $|F| \le k$ such that deleting $F$ from $G$ yields an acyclic graph.
	
	We construct an instance of the stable matching problem with one-sided uncertainty and a matching $M$ as follows:
	\begin{enumerate}
		\item  For each $v \in V$, we add an agent $a_v$ to set $A$ and a matching partner $M(a_v)$ to set $B$.
		\item For each $v \in V$ and $u \in V\setminus \{v\}$, we set $M(a_v) \prec_{a_v} M(a_u)$ if $(v,u) \in E$ and $M(a_u) \prec_{a_v} M(a_v)$ otherwise.
		\item For each $v \in V$ and $u \in  V \setminus \{v\}$, we set $a_v \prec_{M(v)} a_u$.
	\end{enumerate}
	
	Based on the $A$-side preferences, each $(a_v,M(a_u))$ with $(v,u) \in E$ could be a rotation edge and each $(a_v,M(a_u))$ with $(v,u) \notin E$ is not a rotation edge.
	Consider the directed graph $G' = (A\cup B, E')$ with $E' = \{(M(a_v),a_v) \mid v \in V\} \cup \{(a_v,M(a_u)) \mid (v,u) \in E)\}$. Then, based on the $A$-side preferences, each cycle in $G'$ could be a rotation. Furthermore, if we contract the edges $\{(M(a_v),a_v) \mid v \in V\}$, we arrive at the given graph $G$.
	
	Assume that there is a set $F \subseteq V$ with $|F| \le k$ such that deleting $F$ from $G$ yields an acyclic graph. Consider the queries $\topq(M(a_v),A)$ for all $v \in F$. By the third step of the reduction, these queries prove that the agents $M(a_v)$ with $v \in F$ are not part of any rotation.
	This also means that the agents $a_v$ with $v \in F$ cannot be part of a rotation. Thus, the only edges that can still be part of a rotation are the matching edges $(M(a_v),a_v)$ with $v \notin F$ and the edges
$(a_v,M(a_u))$ with $(v,u) \in E$ but $v,u \notin F$. If we consider the graph induced by these remaining edges and contract the matching edges, we arrive at the subgraph $G[V\setminus F]$ of the given feedback vertex set instance. Since this graph by assumption does not contain a cycle, this implies that executing the queries proves that the constructed instance has no rotation.
	
	Consider a query strategy that proves the constructed instance to not have a rotation by using at most $k$ queries. Let $A' \subseteq A$ denote the set of all agents that are returned as the top choice by at least one of those queries. Then, by construction, the alternative query strategy that queries $\topq(M(a_v),A)$ for each $a_v \in A'$ must also be feasible and uses at most $k$ queries. This alternative strategy proves that there exists no rotation by proving that no $a_v \in A'$ is part of any rotation. Thus, removing all vertices $a_v$ and $M(a_v)$ with $a_v \in A'$ from the graph $G'$ as defined above yields a graph without cycles. This also implies that removing $F = \{v \in V \mid a_v \in A'\}$ from $G$ yields a graph without cycles. Thus, $F$ with $|F| \le k$ is feasible for the given feedback vertex set instance.
\end{proof}

In the instances constructed within the proof, querying $\topq(M(a_v),A)$ for between $n-1$ and $n$ agents $a_v \in A$ proves that the given matching is stable and $B$-optimal.
If $n-1$ queries suffice, then this is optimal by~\Cref{lem:bopt:numberlb:sq}. 
Otherwise, $n$ queries are optimal. 
We can decide whether $n-1$ queries suffice via
enumerating all possible choices of the agent $a_v$ for which $M(a_v)$ does not receive a query $\topq(M(a_v),A)$.

Thus, the NP-hardness for proving that no rotation exists does not directly translate to the offline problem of proving that a given matching has no rotation \emph{and is stable}.

    \section{Open Problems}
    \label{sec:conc}
    
    While we understand the \comparison model quite rigorously, it remains open in the set query model what best possible competitive ratio can be achieved for finding a ($A$- or $B$-optimal) stable matching. 
    Further, it would be interesting to investigate the two-sided stable matching problem with uncertainty in the preference lists on both sides further. For verifying the stability of a given matching in this case, we have given a best possible $2$-competitive algorithm. All other questions regarding finding a stable or stable and optimal matching remain open under two-sided uncertainty.
    It would also be interesting to investigate a generalized set query model in which a query to a set $S\subseteq A$ for a $b\in B$ reveals the top-$k$ partners of $b$, that is, the $k$ partners in $S$ that $b$ prefers most.
    

\end{document}